\documentclass[conference,onecolumn]{IEEEtran}
\IEEEoverridecommandlockouts
\PassOptionsToPackage{hyphens}{url}\usepackage{hyperref}
\usepackage{cite}
\usepackage{booktabs,url,enumitem}
\usepackage{amsmath,amssymb,amsfonts,amsthm}
\usepackage{algorithmic,cleveref,csquotes}
\usepackage{graphicx,pgfplots,pgfplotstable}
\pgfplotsset{compat=1.11}
\usepackage{textcomp}
\usepackage{xcolor}
\usepackage{tikz}
\usepackage{lipsum}
\usepackage{upgreek}
\usepackage[caption=false]{subfig}
\usetikzlibrary{automata,arrows,positioning,shapes}
\usetikzlibrary{decorations.text}
\usetikzlibrary{arrows.meta, shapes, positioning}
\def\BibTeX{{\rm B\kern-.05em{\sc i\kern-.025em b}\kern-.08em
    T\kern-.1667em\lower.7ex\hbox{E}\kern-.125emX}}

\renewcommand{\(}{\left(}
\renewcommand{\)}{\right)}

\newtheorem{theorem}{Theorem}

\newtheorem{lemma}[theorem]{Lemma}

\theoremstyle{definition} 
\newtheorem{definition}[theorem]{Definition}

\newtheorem{remark}[theorem]{Remark}

\newcommand{\epoch}{\ensuremath{e}\space}

\newcommand{\h}{h}
\newcommand{\qt}{\enquote}

\newcommand{\checkpointa}{\textnormal{s}}
\newcommand{\checkpointb}{\textnormal{t}}

\newcommand{\signature}{\ensuremath{\mathcal{S}}\space}

\newcommand{\rewardfactor}{\rho}
\newcommand{\baseinterest}{\gamma}

\newcommand{\validatorset}{V}
\newcommand{\basepenalty}{\beta}

\newcommand{\collectivereward}{C}
\newcommand{\deposit}{D}

\newcommand{\voteidc}{{\bf 1}}
\newcommand{\epochlength}{l}
\newcommand{\vldr}{\upnu}
\newcommand{\nvldr}{\tilde\upnu}

\newcommand{\vdsize}{\alpha}
\newcommand{\vtsize}{\mu}

\newcommand{\depsizefactor}{p}

\newcommand{\N}{\mathbb{N}}

\makeatletter
\g@addto@macro\normalsize{%
  \setlength\abovedisplayskip{4pt}
  \setlength\belowdisplayskip{4pt}
  \setlength\abovedisplayshortskip{4pt}
  \setlength\belowdisplayshortskip{4pt}
}

\begin{document}
\title{Incentives in Ethereum's Hybrid Casper Protocol\thanks{This work was supported in part by the National Research Foundation (NRF), Prime Minister's Office, Singapore, under its National Cybersecurity R\&D Programme (Award No. NRF2016NCR-NCR002-028) and administered by the National Cybersecurity R\&D Directorate. Georgios Piliouras acknowledges SUTD grant SRG ESD 2015 097, MOE AcRF Tier 2 Grant 2016-T2-1-170 and NRF 2018 Fellowship NRF-NRFF2018-07.}
}

\author{
\IEEEauthorblockN{Vitalik Buterin\IEEEauthorrefmark{1}, Dani\"el Reijsbergen\IEEEauthorrefmark{2}, Stefanos Leonardos\IEEEauthorrefmark{2},  Georgios Piliouras\IEEEauthorrefmark{2}}
    \IEEEauthorblockA{\IEEEauthorrefmark{1}Ethereum Foundation}
    \IEEEauthorblockA{\IEEEauthorrefmark{2}Singapore University of Technology and Design}
}

\maketitle

\begin{abstract} 
We present an overview of \emph{hybrid} Casper the Friendly Finality Gadget (FFG): a Proof-of-Stake checkpointing protocol overlaid onto Ethereum's Proof-of-Work blockchain. We describe its core functionalities and reward scheme, and explore its properties. Our findings indicate that Casper's implemented incentives mechanism ensures liveness, while providing safety guarantees that improve over standard Proof-of-Work protocols. Based on a minimal-impact implementation of the protocol as a smart contract on the blockchain, we discuss additional issues related to parametrisation, funding, throughput and network overhead and detect potential limitations.
\end{abstract}


\section{Introduction}\label{sec:introduction}

In 2008, the seminal Bitcoin paper by Satoshi Nakamoto \cite{nakamoto2008bitcoin} introduced the \emph{blockchain} as a means for an open network to extend and reach consensus about a distributed ledger of digital token transfers.
The main innovation of Ethereum \cite{buterin2014next} was to use the blockchain to maintain a history of code creation and execution.
As such, Ethereum functions as a \emph{global computer} that executes code uploaded by users in the form of \emph{smart contracts}.
Like Bitcoin \cite{Ga15,Ga17}, Ethereum's block proposal mechanism is based on the concept of Proof-of-Work (PoW).
In PoW, network participants utilise computational power to win the right to add blocks to the blockchain.
However, the alarming global energy consumption of PoW-based blockchains has made the concept increasingly controversial \cite{Sw18,digiconomist,Pre20}.
One of the main alternatives to PoW is \emph{virtual mining} or \emph{Proof-of-Stake} (PoS) \cite{Bon15,bentov2016pos,firstproofofstake,Vot20}. In PoS, the right to propose a block is earned by locking -- or \emph{depositing} -- tokens on the blockchain, which has no inherent energy cost. \par
As part of its long-term goal to switch from PoW to PoS, Ethereum is designing a full PoS protocol called \emph{Casper} \cite{ryan2018casper,ryan7,Bu1810}. To ensure a smooth transition with minimal impact on its users and the Ether price (Ethereum's native cryptocurrency \cite{Kok20}), Ethereum deployed and tested a \emph{hybrid} version, \emph{Casper the Friendly Finality Gadget} or Casper FFG, \emph{as a smart contract} on a dedicated testnet \cite{caspercontract,Pe17,Pe18}. Essentially, Casper FFG is a simplified version of a Byzantine fault tolerant protocol \cite{castro1999practical}, with \qt{votes} for checkpoints taking the place of \emph{prepares} and \emph{commits}. In contrast to protocols that treat every block as a checkpoint (e.g., Tendermint \cite{kwon2014tendermint}
and Algorand \cite{chen2016algorand}),
 Casper FFG periodically checkpoints blocks on an underlying chain. As such, the tried and tested PoW chain can be preserved during a transitional phase, whilst the extra load on the network is mitigated. This addresses two of the classical challenges that affect PoS protocols \cite{Li17,Br18}: the \emph{nothing-at-stake} problem through the \emph{slashing} mechanism that penalises misbehaving violators \cite{Bu1811}, and \emph{long-range attacks} through a modified \emph{fork-choice rule} that prioritises (and never automatically reverts) finalised checkpoints over PoW \cite{ryan2018eip}.\par
%Since Ethereum is an established PoW blockchain with a large user base, gradual change is preferable -- hence, before moving to full PoS, there will be a \qt{hybrid} phase with PoW and PoS existing simultaneously.
%One element of the Casper protocol that is particularly well-suited to a transitional stage is the \emph{Friendly Finality Gadget} (FFG), a mechanism that can be overlaid on the Ethereum blockchain with minimal impact to its users \cite{caspercontract}.
%We note here that the Casper FFG contract was implemented and evaluated on a dedicated testnet, see \cite{}.
%It provides \emph{finality}: i.e., when a node accepts a block as finalised, it will never overturn it (except after a manual reboot).
%This finality guarantee is strong: 
The high-level idea and fundamental properties of the hybrid Casper FFG have been outlined in \cite{Bu1811}. In the present paper we extend \cite{Bu1811} to include a full description of the implemented incentives (reward--penalty) scheme. We show that the scheme is incentive-compatible in the sense that participants are incentivised to follow the protocol, and we investigate its impact on the basic properties of liveness and safety. Based on the minimal-impact implementation of Casper FFG as a smart contract on the PoW chain, the present paper also covers the parameter choice, confirmation times, and network overhead, and a discussion of potential limitations. Along with the lessons learned, we present some directions, ideas, and design concepts for future development of Ethereum's consensus mechanism and blockchains in general.\par
The contributions of this paper are as follows. We first provide an overview of the Casper FFG protocol and describe its core functionalities. To reason about liveness and safety, we develop a mathematical framework for the incentives scheme, slashing conditions, and the fork-choice rule. Our first theoretical result is that with the implemented reward scheme, Casper's checkpointing scheme is $\alpha$-live, for any $\alpha \in \(0,1\right]$, i.e., online validators controlling any fraction $\alpha\in \(0,1\right]$ of the stake will be eventually able to finalise checkpoints. Concerning safety, we first show that the property proved in \cite{Bu1811} carries over to the updated incentives scheme, namely that two or more \qt{conflicting} checkpoints can only exist if validators controlling at least $1/3$ of the stake misbehave conspicuously and hence, can be slashed (i.e., punished). We then extend the safety analysis to the case of a protracted network partition. It is clear from the CAP theorem \cite{gilbert2002brewer} that during a network partition, a distributed ledger can satisfy liveness or safety, but not both. We indeed find that the liveness guarantee has implications for safety, as it eventually allows for conflicting checkpoints to be finalised. However, using both numerical and analytical tools, we derive that the minimum duration of such a network partition is very large (i.e., at least three weeks). As such, Casper's checkpointing protocol prioritises safety in the short term, but liveness in the long term, and hence strikes a balance between protocols that either always prioritise liveness (e.g., the underlying PoW chain) or safety (e.g., Tendermint). Next, we study the incentive compatibility of the protocol by investigating whether it is profitable for nodes to follow the protocol. Finally, we turn to the implementation of Casper FFG as a PoW chain contract. In this approach, the stakeholders' deposits are stored as variables in the contract, and their actions appear on the blockchain as contract calls. In \Cref{sec:implementation}, we evaluate the effect on transaction throughput and computational cost (or load) as measured via \emph{gas}, discuss the impact of different parametrisations, and identify potential limitations. \par
To remain compatible with Ethereum's evolution towards a \emph{sharded} and hence more scalable design \cite{ryan2018casper,buterin2018cbc,shardingroadmap}, the specifications of Casper FFG are constantly updated \cite{Bu18,Bu1810}. Yet, the key takeaways and the main findings from testing Casper FFG as a smart contract in a hybrid setting carry over to the currently developed design of a main PoS chain, called the \emph{Beacon Chain} \cite{Bu1810,buterin2018beacon}, that will coordinate consensus among several side chains or \emph{shards}. In particular, the main components, i.e., the incentives mechanism, functionality, and design philosophy remain basic components of Casper FFG even in the sharded construct (i.e., multi-chain) setting \cite{buterin2018beacon,edgington2018,Bu1810}. Although we focus on combining Casper FFG with Ethereum's PoW chain, the protocol can be overlaid on top of any chain-based blockchain -- PoW or PoS \cite{ryan2018eip} -- and may therefore be of wider interest to the blockchain community. 

\subsubsection*{Outline} In \Cref{sec:casper}, we provide an abstract overview of Casper FFG and its operations in a formal mathematical model. \Cref{sec:analysis} contains our main theoretical results on liveness, safety, and incentive compatibility. In \Cref{sec:implementation}, we present our findings on Casper's implementation and discuss related issues. We summarise our results in \Cref{sec:conclusions}.
  
%As we discuss in the paper, the protocol can be implemented with minimal impact on the existing Ethereum blockchain by encoding it in a smart contract.  \Cref{sec:conclusions} concludes the paper with directions for future work.

\section{The Hybrid Casper FFG Protocol}\label{sec:casper}

\subsection{The PoW Chain}\label{sub:powchain}

Ethereum functions as a global computer whose operations are replicated across a peer-to-peer network. 
Participants in the network are called \emph{nodes} -- they typically interact with the rest of the network via a software application called a \emph{client}.
The client interacts with the Ethereum blockchain via \emph{transactions}. There are three main types of transactions:
\begin{description}[leftmargin=0cm]
\item[Token transfers] provide the same core functionality as Bitcoin by allowing nodes to exchange digital tokens.
\item[Contract creations] upload pieces of code, called \emph{(smart) contracts}, to the blockchain. Contracts are executed using a runtime environment called the Ethereum Virtual Machine (EVM).\footnote{The eWASM framework \cite{ewasm} may replace the EVM in the future.} Two notable high-level languages that compile into the EVM are \href{https://solidity.readthedocs.io}{Solidity} and \href{https://github.com/ethereum/vyper}{Vyper} which are based on the JavaScript and Python languages, respectively. Vyper is particularly relevant for this paper as it is used for the Casper contract. A typical contract will include one or more \emph{functions} which can be called by the nodes. 
\item[Contract calls] handle interactions with the functions of an existing contract\footnote{For more details on these notions see \Cref{sub:messages}.}.
\end{description}
In PoW, the ordering of these transactions in the blockchain is determined by a special class of nodes called \emph{miners}. Miners collect and sort transactions, after which they execute these transactions in the EVM. Upon completion, information about the state of the complete global computer (account balances, contract variables, etc.) is then combined with the transactions and various other data into a data structure called a \emph{block}. Miners compete to find a block that satisfies a condition that requires considerable computational effort. The winning miner receives a fixed amount of Ether (ETH) -- Ethereum's native currency -- in the form of a mining reward. Additionally, all of the three transaction types listed above require \emph{gas}, which is also paid to the winning miner as a reward for the computational effort. In particular, more computationally expensive operations in the EVM require more gas (see \cite{wood2014ethereum} for a complete specification of the gas cost for the different types of operations) -- as such, gas cost is a good measure for the computational `cost' of a transaction. In Section~\ref{sec:implementation} we will investigate the gas cost of the different functions in the Casper contract.

\subsubsection*{Formal Framework}
To reason about the evolution of the PoW chain in the context of network latency and partitions, we require a formal framework. Let $N$ denote the full set of nodes as identified by an integer denoting their index in the network, i.e., $N \subseteq\mathbb N$. At each time $t\geq0$, each node $n\in N$ is aware of a set of blocks $B$ which we denote by $\mathcal B_n^t$, i.e.,
\begin{align*}
\mathcal B^t_n:=\{& B:\text{node $n$ is aware of $B$ at time $t\ge0$}\}.
\end{align*}
The \emph{genesis} block $g$ is the only block that all nodes are aware of at time $0$, i.e., $\mathcal{B}^0_n=\{g\}$ for all $n\in N$. Each block $B$ can be represented uniquely as an integer, e.g., via the hash of its header, although this means that the range of possible blocks is \qt{restricted to} $\{0, \ldots, 2^{256}-1\}$. 
%Given a set of known blocks, each mining node will need to decide which block to try to extend -- this block is called the \emph{head}. To decide which block is the head, honest nodes follow a \emph{fork-choice rule}, which can be seen as a preference relation among blocks.
Due to network latency, eclipse attacks, or other reasons, it may be the case that different nodes are aware of different sets of blocks, i.e., there may exist $t>0$ and nodes $n,m\in N$ such that $\mathcal B^t_n\neq \mathcal{B}^t_m$. We assume that nodes cannot \emph{forget} about blocks, i.e.,
\[\mathcal B_n^t\subseteq \mathcal B_n^s, \quad \text{for any $s,t>0$ with $s>t$}.\]
Each block $B$ points to a previous block $P\(B\)$ via a function \mbox{$P:\N \rightarrow \N$}, cf. \cite{Br18}, with 
\begin{itemize}[leftmargin=*]
\item $P^0\(B\):=B$ for any block $B$,
\item $P^2\(B\) = P\(P\(B\)\), P^3\(B\) = P\(P\(P\(B\)\)\)$ etc. 
\item $P^k\(g\):=\emptyset$ for all $k\ge1$, where $g$ is the genesis block,
\item $P^k\(B\)\neq B$ for any $B$ and $k\in \{1,2,\ldots\}$, i.e., there are no cycles.\footnote{An attacker could theoretically construct a cycle, but the probability of succeeding is negligible: i.e., after picking a previous hash, the next block would need to have that exact hash, which occurs with probability $<10^{-70}$.}
\end{itemize}
Based on this relationship, we define the \emph{chain $C\(B\)$ of a block $B$} as the path from $B$ leading back to $g$, i.e., $C\(B\):=\(B,P\(B\),\dots, P^{k-1}\(B\),g\)$. If $B'\neq B$ is a block such that $B'\in C\(B\)$ then $B'$ is called an \emph{ancestor} of $B$ and $B$ is called a \emph{child} of $B'$. If $B'=P\(B\)$, then $B$ is called a \emph{direct} child of $B'$. The length $k$ of $C\(B\)$ determines the \emph{block height $h\(B\)$} of block $B$, i.e., $h\(B\)=k$ if $P^k\(B\)=g$. The height of the genesis block is $0$, i.e., $h\(g\)=0$. The \emph{state of the blockchain} in block $B$ is obtained by executing all transactions in $C\(B\)$ starting from the genesis block $g$, see also \cite{Leo16}.

A \emph{fork} occurs whenever two different blocks $A$ and $B$ exist such that $P(A) = P(B)$. At any time $t$, each miner $n$ needs to decide which block in $\mathcal{B}^t_n$ to extend. This is done using a \emph{fork-choice rule}, which in its simplest form is a function $f$ that maps a set $\mathcal{B}$ to a single block $B$. The block chosen is called the \emph{head}. In PoW, the fork-choice rule is to select the block $B$ with the highest accumulated proof-of-work, and in case of ties to prefer the block seen first. In the absence of serious attacks or network partitions \cite{Ey14,Sa17,Zu18}, a naive measure of a chain's PoW is the height $h\(B\)$ of a block. Hybrid Casper FFG's fork-choice rule is discussed in the next section.

%One of the information types contained in each block is a reference to a predecessor. Formally, let The only block that does not have a predecessor is the first block, which is called the \emph{genesis}. For all other blocks, the sequence of blocks leading back to the genesis is called the \emph{chain} of that block.

\subsection{Execution of Hybrid Casper FFG}\label{sub:execution}

\subsubsection*{Validators}

In hybrid Casper FFG, some nodes assume the role of \emph{validators}. Nodes can become validators by locking/staking tokens on the PoW chain, thus creating a \emph{deposit}. In the Casper contract, this is done by calling the {\tt deposit} and {\tt withdraw} functions. Additionally, the {\tt logout} function removes a validator from the active validator set and needs to be called before withdrawing. Validators need to wait a long period after depositing before being allowed to withdraw. In the benchmark setting \cite{ryan2018eip}, this is 15000 epochs, i.e., around 120 days.
%Although we will occasionally comment on the impact of changes to the validator set on our analysis (particularly in Section~\ref{sec:messages}), we consider the validator set to remain static for the experiments in this paper. 
%Hence, to simplify notation, we assume that there is a static validator set, denoted by $\cvalidatorset$, and refer the reader to Section 3 of \cite{buterin2017casper} for more detail on dynamic validator sets.

\subsubsection*{Checkpoints}

The central role of the validators is to \emph{vote} for checkpoints\cite{Da18}. A \emph{checkpoint} is any block with block number $i \cdot \epochlength$, where $i \in \{0,1,\ldots\}$ and $\epochlength \in \mathbb{N}$ denotes the \emph{epoch length}: an \emph{epoch} is defined as the contiguous sequence of blocks between two checkpoints, including the first but not the latter. Block $0$ (which is also a checkpoint) denotes the genesis block. We will assume $\epochlength = 50$ throughout this paper \cite{ryan2018eip}. The \textit{epoch of a block} is the index of the epoch containing that hash, e.g., the epoch of blocks 200 and 249 is 4. Because validators' deposits change between epochs due to the rewards and penalties we denote by $\deposit_{\vldr,i}$ the deposit of validator $\vldr \in V$ at the beginning of epoch $i\in\mathbb N$, where $V$ denotes the set of \emph{active validators} at epoch. Deposits also change between blocks but since withdrawals and deposits only occur at the end of epochs, it is sufficient to track the deposits at those points. In practice, the validator set is dynamic and changes over time too. For the present analysis it suffices to assume a fixed validator set. We denote by $\deposit_i:=\sum_{\vldr\in \validatorset}\deposit_{\vldr,i}$ the total stake (deposited value) at epoch $i$.

\subsubsection*{Votes in Casper}
In the Casper contract, voting is done by calling the {\tt vote} function with the arguments $\langle \vldr, \checkpointb,  \h(\checkpointb)$, $\h(\checkpointa), \signature \rangle$, where the entries are explained in Table~\ref{tbl:vote}.
\begin{table}[!ht]
\centering
\begin{tabular}{@{}l l@{}}
	\midrule
	\textbf{Notation} & \textbf{Description} \\ \midrule
	$\vldr$ & validator index \\
	\checkpointb & hash of the `target' checkpoint \\
	\h(\checkpointb) & height of the target checkpoint in the checkpoint tree \\
	\h(\checkpointa) & height of the source checkpoint in the checkpoint tree \\
	\signature & signature of $( \vldr, \checkpointb,  \h(\checkpointb), \h(\checkpointa) )$ by the validator's private key \\
	\midrule
\end{tabular} 
\vskip0.1cm
\caption{Summary of {\tt vote} function arguments.} \label{tbl:vote}
\end{table}
We will say that a validator $\vldr$ who generates a vote transaction to the Casper address votes \emph{correctly} or casts a \emph{valid} vote if $\vldr$ includes the expected source and target epochs (checkpoints) returned to $\vldr$ by a Casper contract call and a valid signature created by $\vldr$'s private key \cite{Da18,ryan2018eip}. In particular, the only valid target epoch is the current epoch (in the contract) and the only valid target for a vote that appears in block $B$ is the checkpoint block with the correct height \emph{on the chain} $C\(B\)$. Hence, the correctness of a vote message depends on the chain that a node is voting on.\footnote{We address this case in more detail in \Cref{sec:analysis}.}

\subsection{Finalisation}\label{sub:finalisation}
A {checkpoint $\checkpointb$} is \emph{justified} if at least $2/3$ of the validators -- in terms of stake -- publish a vote of the form $\langle \vldr, \checkpointb, \text{\h(\checkpointb), \h(\checkpointa)}, \signature \rangle$ for some checkpoint $\checkpointa$, such that $\checkpointa$ is an ancestor of $\checkpointb$ and is itself justified. A checkpoint $\checkpointa$ is \emph{finalised} on a chain if it is justified and the checkpoint $\checkpointb$ that immediately follows $\checkpointa$ on that chain is also justified -- i.e., if $\checkpointb = P^{\epochlength}\(\checkpointa\)$ and $\checkpointa,\checkpointb$ are both justified, then $\checkpointa$ is finalised. We consider $g$ to be both justified and finalised, to serve as the base case for the recursive definition. Two finalised checkpoints $\checkpointa, \checkpointb$ are \emph{conflicting} if neither is an ancestor of the other. %At this point, the current and previous validator sets are assume not to be empty.
When handling vote transactions, the Casper contract only processes (i.e., includes) \emph{correct} votes as defined above. Votes for the \emph{wrong} targets and target epochs are considered invalid and ignored. A typical execution of the Casper protocol is illustrated in \Cref{fig:variables}.
\begin{figure*}[!htb]
\centering
\includegraphics{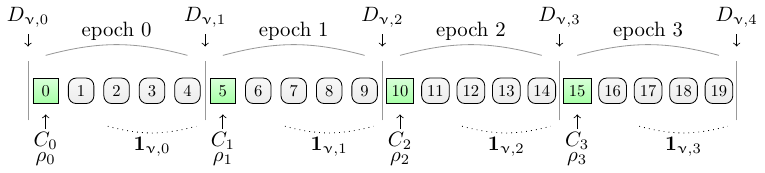}
\caption{Illustration of the points in the blockchain where the different variables are updated. Blocks are mined by miners on the underlying PoW blockchain. Validators vote on checkpoints -- which are highlighted in green -- in the later part of the epoch (which have length $l=5$ for illustrative purposes). Correct votes form a link between the previous two checkpoints. For instance, correct votes in epoch 3 have as source checkpoint $10$ and target checkpoint $15$.}
\label{fig:variables}
\end{figure*}
The notion of finality is a central primitive in the design of Casper and before proceeding with the rest of the definitions, we elaborate on its relationship with the typical notions of finality in PoW blockchains. 

\subsubsection*{Finality in Proof of Work}
Technically, a PoW blockchain never allows a transaction to truly be \qt{finalised}: for any given block $B$ with chain $C\(B\)$, there is always a positive probability that a miner will create a longer chain that starts from an earlier block $B'$, i.e., $B'\in C\(B\)$ with $B'\neq B$, and which does not include $B$. However, although the probability is positive, it decreases exponentially in the distance $k$ between $B$ and $B'$, i.e., as expressed via $B'=P^k\(B\)$ \cite{Ga15}. This argument was first made in the Bitcoin white paper based on the solution to the Gambler's Ruin problem \cite{Na08}. In particular, the chain of an attacker who controls less than a fraction $\alpha$ of the network hash power (mining power), for any $\alpha \in (0,0.5)$, can be shown to have a probability of at most $\(\alpha/\(1-\alpha\)\)^k$ to overtake the original chain if the original chain is $k$ blocks longer when the attacker launches the attack. By choosing a sufficiently high $k$, the probability of reverting a block can be reduced to practically theoretical levels \cite{Si16}, cf. \Cref{fig:pow_finality}.
\begin{figure}[!htb]
\centering
\includegraphics{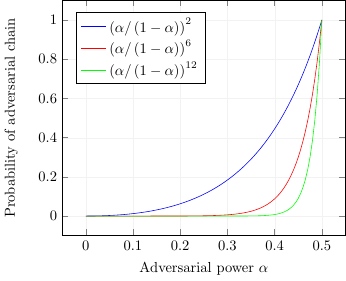}
\caption{Probability that the chain of an adversary controlling a fraction $\alpha$ of the total mining power will overtake a chain of length $k$ for $k=2,6,12$. The probability quickly becomes negligible as $k$ increases, yielding strong probabilistic finality even in PoW chains.}
\label{fig:pow_finality}
\end{figure}
Hence, a strong notion of finality can be argued to exist even on PoW blockchains \cite{Si16}. However, this argument assumes that more than half of the nodes behave honestly, and that they are coordinated on the same chain. There are a multitude of scenarios in which this assumption is invalid, e.g., due to high network latency, eclipse attacks \cite{heilman2015eclipse}, or the sudden influx of (possibly malicious) mining power \cite{Bo15,bonneau2018hostile}. The latter scenario is a particularly big concern for smaller PoW blockchains, which can be overwhelmed by mining power that is normally used on bigger blockchain platforms. One recent example is the attack on Ethereum Classic that occurred in January 2019 -- according to engineers at Coinbase, a prominent cryptocurrency exchange, one block with as many as 122 confirmations was overturned \cite{nesbitt2019}. We can therefore conclude that even though PoW's probabilisitic finality provides strong guarantees under normal conditions, such conditions cannot always be guaranteed in practice.
%In practice, different users have different thresholds depending on how risk-averse they are -- e.g., 6 confirmations is a rule of thumb recommended by the Bitcoin Wiki, whereas Coinbase requires 3 confirmations for Bitcoin \cite{coinbaseconfirmations}, 35 for Ethereum \cite{coinbaseconfirmations}, and either 3527 \cite{coinbaseconfirmations} or 5676 \cite{coinbaseetc} confirmation for Ethereum Classic.

\subsubsection*{Finality in Casper FFG}\label{sub:finality_casper}
The Casper protocol is intended to offer stronger finality guarantees than PoW in both the hybrid PoW/PoS and ultimately in the pure PoS settings. The key insight is that in PoS, nodes must make a deposit to become validators, and their messages that appear on the blockchain can therefore be linked to those deposits. If the users engage in clear misbehaviour, e.g., by voting for conflicting checkpoints, then they can be punished by slashing their deposits. The equivalent notion in PoW would be that the mining hardware of adversarial miners is destroyed. In our setting, the worst case of misbehaviour is to vote for conflicting checkpoints. If $2/3$ of validators are placing their stakes behind their votes for a checkpoint, then if another $2/3$ of validators place their stakes behind a contradictory checkpoint, then that necessarily implies that the intersection, i.e. at least $1/3$ of validators, have been backing both of the conflicting checkpoints. We still cannot guarantee absolute finality -- if conflicting checkpoints are finalised due to validator misbehaviour, then either the chain permanently forks or some off-chain governance mechanism is used to decide in favour of one branch, at the expense of the other. However, we have an \textit{economic} form of finality that guarantees that a finalised block is never overturned \textit{under the condition that} no group of validators is willing to suffer a financial loss of at least $1/3$ of the total value of the staked deposits. In PoW, it is impossible to put concrete bounds on the required financial cost to overturn a finalised block, although estimates can be made based on the cost to rent mining power \cite{bonneau2018hostile}.

As we will see in Section~\ref{sec:analysis}, there is one other scenario in which finality in Casper FFG can theoretically be compromised, namely due to liveness recovery in more than one branch during a serious, protracted network partition. However, we will show that the minimum required time duration of such a partition is extremely long, namely over three weeks. In reality, such a partition would require that a significant fraction of validators is not able to connect to the rest of the network, even via satellite, for several weeks, yet that they are able to communicate amongst themselves to vote on each other's blocks. Even in such an unprecedented case, it may be possible to retroactively compensate deposit holders on the losing branch through a hard fork coordinated by an off-chain governance mechanism (e.g., as in the case of the DAO hard fork). Hence, we do not consider this scenario to be a realistic threat to block finality.

%\qt{economic finality}: while there cannot be a guarantee that a \qt{state $X$ will never be reverted}, there can be the slightly weaker guarantee that \qt{either state $X$ will never be reverted or a large group of validators will be accountable -- and hence, in economic terms, penalisable -- for violating a pre-determined behaviour}. Formally, a block $B$ is \emph{economically finalised}, with crypto-economic security margin $X$, if a client has proof that either (i) $B$ is going to be part of the canonical chain forever, or (ii) those actors that caused $B$ to get reverted are guaranteed to be economically penalised by an amount equal to at least $X$. These notions are analysed in detail in the subsequent sections.

%Note that due to forks, different blocks may exist on the network such that a single vote is considered valid on one block's chain and not on the others. 

\subsubsection*{Fork-Choice Rule}
The Hybrid Casper FFG fork-choice rule extends the PoW fork-choice rule of Section~\ref{sub:powchain}. In particular, to select a block as the head of the chain, a client queries the contract to find the block(s) with the \emph{highest justified epoch}, prioritising the block with the \emph{highest mining PoW} only in a case of tie. Clients only consider epochs in which the total deposit exceeds a given minimum threshold and never revert a \emph{known finalised} checkpoint. If a chain has no justified checkpoints beyond $g$, the fork-choice rule essentially reverts to pure PoW \cite{ryan2018eip}.

\subsection{Rewards and Penalties Mechanism}\label{sub:rewards}

%\begisn{figure*}[!h]
%\centering
%\livenessfig
%\caption{Liveness!}
%\label{fig:variables}
%\end{figure*}
%
%\begin{figure*}[!h]
%\centering
%\safetyfig
%\caption{Safety?}
%\label{fig:variables}
%\end{figure*}

According to Casper's payments scheme, validators who vote correctly during an epoch are rewarded, and validators who do not are penalised. This is achieved by adjusting the validators' deposits depending on their own voting behaviour, i.e., on whether they are voting or not, and on the performance of the protocol as a whole, i.e., on what total fraction of validators vote correctly and whether that is enough to justify and finalise checkpoints. Specifically,
\begin{description}[leftmargin=0cm]
\item[Correct voting:] If checkpoints are being finalised, i.e., if at least $2/3$ of validators are voting correctly, then deposits of correctly voting validators increase by a positive interest rate. If checkpoints are not being finalised, deposits of correct voting validators remain the same. The interest rate depends on the total deposit and how many other validators are voting. 
\item[Non-voting:] Regardless of finalisation, non-voting validators are penalised and their deposits shrink. The penalties are increasing in the proportion of non-voting validators. If epochs fail to be finalised during sustained time periods, then the penalties gradually become  more severe. 
\item[Conflicting/incorrect voting:] Incorrect votes are ignored and validators who cast them are considered as non-voters and are not rewarded. If evidence is provided that validators cast conflicting votes, then their deposits are partially or entirely removed (slashed) depending on the severity of the violation and the overall protocol performance (i.e., the proportion of \qt{correct} voters and whether epochs are being finalised or not). 
\end{description}
Incorrect or missing votes are not punished as harshly as conflicting votes, as there are several faulty behaviours that can cause a validator $\vldr$ to fail to produce a valid vote, as we discuss below.
\begin{enumerate}
	\item A valid vote by $\vldr$ for the target epoch either never makes it onto the blockchain, or only outside the target epoch. This could be because
	\begin{enumerate}[noitemsep]
		\item $\vldr$ never cast the vote, or was too late when doing so, \label{it:novote}
		\item a majority of the other validators are censoring $\vldr$, \label{it:validators}
		\item a significant portion of the nodes decided not to propagate $\vldr$'s vote, or \label{it:nodes}
		\item the network latency was too severe. \label{it:latency}
	\end{enumerate}
	\item A vote by $\vldr$ does make it onto a block inside the target epoch, but it is otherwise invalid. This could be because
	\begin{enumerate}[noitemsep]
		\item the signature is invalid, \label{it:signature}
		\item the source epoch does not use the last justified epoch, or \label{it:source}
		\item the target hash is invalid. \label{it:target}
	\end{enumerate}
\end{enumerate}

In case \ref{it:novote}, the fault clearly lies with $\vldr$. However, in cases \ref{it:validators} and \ref{it:nodes} the fault lies with (a coalition of) malicious validators or other nodes respectively. In case \ref{it:latency}, the fault does not even necessarily lie with participants in the blockchain. In case~\ref{it:signature}, this could have been due to a fault by $\vldr$, or by one or more malicious or malfunctioning nodes tampering with $\vldr$'s message before propagating. In case~\ref{it:source}, this is either the result of a mistake by $\vldr$, a safety fault (in which a previously finalised epoch is overturned), or because $\vldr$ had an incomplete view of the network. In case~\ref{it:target}, the blockchain must have been forked when $\vldr$ voted --- again, this could have been caused by the network, a mistake by $\vldr$, misbehaving nodes (who blocked $\vldr$'s view of the network), misbehaving miners (who were withholding blocks from $\vldr$), or a misbehaving coalition of validators (who finalised a chain that was different from the one preferred by the network).\par
The difficulty with fault attribution in the case of missing/incorrect votes (which in general suffers from the problem of \textit{speaker/listener fault equivalence} \cite{buterin2017triangle}) creates a tension between \textit{preventing harm} and \textit{fairness}, i.e., between sufficiently penalising malicious validators and adequately incentivising honest validators to participate. In Casper's implementation, these adjustments are achieved via a scheme of reward factors, cf. \eqref{eq:rewardfactor} and \eqref{eq:collective} that properly update validator's deposits, cf. \eqref{eq:interest}. This payment scheme is designed to make the protocol \emph{incentive-compatible}, i.e., to encourage validators to vote correctly and as often as possible, and enforces the protocol's purposes of \emph{liveness} and \emph{safety}, as we discuss in more detail in \Cref{sec:analysis}. 
%We describe the details of the payment scheme implemented in the Casper contract below.

%In the following, we describe the reward scheme implemented in the contract: all the formulae are mathematical representations of its code. The scheme ensures that validators who vote successfully during an epoch earn `interest': their deposit is increased by some percentage. However, validators who do not vote will earn a penalty in the form of negative interest. 
%The penalties serve two purposes:
%\begin{enumerate}
%\item Incentive compatibility: validators are encouraged to vote as often as possible.
%\item Liveness: when checkpoints stop being finalised due to validators going offline, then the deposits of the non-voting validators will decrease until their share eventually drops below $\frac{1}{3}$ and finalisation can resume.
%\end{enumerate}
%However, due to speaker/listener fault equivalence, the contract ensures that when some validators do not vote, all \emph{other} validators in the network get fewer rewards or higher penalties. Let $m$ denote the (weighted) fraction of validators who voted, and $\rho$ the reward rate (to be discussed in more detail below). Then the deposit of a validator $\vldr$ should be
%\begin{itemize}
%\item multiplied by $1+\frac{1}{2} m \rho $ if $\vldr$ votes during an epoch, and
%\item multiplied by $1-(1-\frac{1}{2} m) \rho $ if $\vldr$ does not.
%\end{itemize}
%In the following, we discuss more concretely how a reward scheme that (roughly) enforces this can be achieved.

\subsubsection*{The Scheme}
\label{sub:scheme}
In detail, the Casper contract adjusts validators' deposits via the \emph{individual} and the \emph{collective} reward factors, $\rewardfactor_i$ and $\collectivereward_i$ respectively, that depend on each validator's voting behaviour and on the aggregate protocol functionality, i.e., on whether epochs are finalised or not, for every epoch $i\in \mathbb N$. Specifically, let ESF$_i$ denote the \emph{Epochs Since Finalisation} during epoch $i$, defined as $i$ minus the height of the last finalised epoch. Since $g$ is finalised, ESF$_0 = 0$ and ESF$_1=1$. For $i\ge2$, ESF$_i$ always equals $2$ in the ideal situation where everyone always votes (it requires two consecutive epochs to be justified for the first of them to be finalised), otherwise it is higher. Based on the total stake $\deposit_i$ in epoch $i\in \mathbb N$, the \emph{individual reward factor} $\rho_i$, is then defined as follows 
\begin{equation}\label{eq:rewardfactor}
\rewardfactor_i :=\begin{cases} \baseinterest\cdot\deposit_i^{-\depsizefactor} + \basepenalty \cdot \(\text{ESF}_i-2\), & \text{if }\deposit_i>0 \\0, &\text{otherwise}\end{cases}
\end{equation}
where $\baseinterest,\depsizefactor,\basepenalty>0$ denote the \emph{base interest}, \emph{total deposit dependence} and \emph{base penalty} factors, respectively. In the benchmark specification\cite{ryan2018eip}, the parameters were set as \mbox{$\baseinterest = 7\cdot10^{-3}$}, \mbox{$\basepenalty=2\cdot10^{-7}$},  and $\depsizefactor = 1/2$. Here, $\depsizefactor$ was chosen to strike a balance between $\depsizefactor=0$ (i.e., constant interest rate per validator) and $\depsizefactor=1$ (i.e., constant total amount of interest paid out by the protocol). In particular, given $\depsizefactor=1/2$, $\baseinterest$ and $\basepenalty$ were derived by reverse-engineering the constants from two desired outcomes: assuming that 10M ETH has been deposited, i) validators earn $\approx$ 5\% annual interest if everyone (nearly) always votes, and ii) if 50\% of the validators go offline, they lose 50\% of their deposits in 21 days\footnote{The parameter choice will be discussed further in Section~\ref{sec:implementation} and \Cref{sub:choice} for more details on their values and choice.}.\par
To define the \emph{collective} reward factor, let
\[\voteidc_{\vldr,i} = \begin{cases} 1 & \text{if $\vldr$ voted successfully during epoch $i$,}\\ 0 &\text{otherwise.}\end{cases}\]
At the end of epoch $i$, let $m_i$ denote the weighted fraction of correctly voting validators in epoch $i$, defined as $m_i := \frac1{\deposit_i}\textstyle\sum_{\vldr \in \validatorset} \voteidc_{\vldr,i}\deposit_{\vldr,i}$. The \emph{collective} reward factor, $\collectivereward_i$, at the beginning of epoch $i$ is defined as
\begin{equation}\label{eq:collective}
\collectivereward_i := \begin{cases} \frac{1}{2} m_{i} \rewardfactor_{i}, & \text{if ESF}_{i} = 2, \\ 0, & \text{otherwise.} \end{cases}
\end{equation}
Validators' deposits are then updated via the scheme
\begin{equation}
\deposit_{\vldr,i+1} :=  (1 + \collectivereward_{i}) \frac{1 + \voteidc_{\vldr,i}\rho_{i}}{1+\rho_{i}}\deposit_{\vldr,i}. \label{eq:interest}
\end{equation}
The intuition behind the updates in validators' deposits $\deposit_{\vldr,i}$ in \eqref{eq:interest} and the factors \eqref{eq:rewardfactor} and \eqref{eq:collective} that are used is the following.
\begin{description}[leftmargin=0cm]
\item[Case 1: ESF$_i>2$.] In this case, epochs are not being finalised properly, i.e., there exists a fraction of more than $1/3$ of validators that are not voting as expected. Accordingly, $C_i=0$ and hence \[\deposit_{\vldr,i+1} = \frac{1 + \voteidc_{\vldr,i}\rho_{i}}{1+\rho_{i}}\deposit_{\vldr,i}.\]
Obviously, the actual value of $\deposit_{\vldr,i+1}$ now depends on the value of $\mathbf 1_{\vldr,i}$, i.e., on whether validator $\vldr$ voted successfully during epoch $i$ or not. 
\begin{description}[itemindent=0cm]
\item[Case 1a: $\mathbf 1_{\vldr,i}=1$.] In this case, validator $\vldr$ voted successfully and $\deposit_{\vldr,i+1}= \deposit_{\vldr,i}$, i.e., validator $\vldr$'s deposit remains the same. In other words, if checkpoints are not finalised but a validator keeps voting successfully, then their deposits remain unharmed. 
\item[Case 1b: $\mathbf 1_{\vldr,i}=0$.] In this case, validator $\vldr$ did not vote successfully which means that they essentially contribute with their faulty behavior in checkpoints not being finalised. According to \eqref{eq:interest}, $\deposit_{\vldr,i+1}= \frac{1}{1+\rho_i}\deposit_{\vldr,i}$, and hence, since $\rho_i>0$, cf. \eqref{eq:rewardfactor}, the non-voting validator $\vldr$'s deposit will shrink. However, the exact intensity depends on the value of $\rho_i$ and the selection of parameters $\gamma,p$ and in particular, by the base penalty factor, $\beta$. Specifically, since ESF$_i>2$, a larger $\beta$ \qt{penalises} the non-voting validator $\vldr$ by increasing $\rho_i$. Additionally, the penalties further increase as ESF$_i$ increases, i.e., as the abnormal network conditions (non-finalization of checkpoints) persist. 
\end{description}
\item[Case 2: ESF$_i=2$.] In this case, checkpoints are finalised as expected and possibly only a minority of validators -- less than $1/3$ in terms of their stake -- are not voting as expected. This implies that $\rho_i=\gamma \cdot D_i^{-p}$, i.e., that the individual reward factor for validator $\vldr$ depends now only on the base interest $\gamma$, the total dependence $p$ and the total staked deposit at epoch $i$, $D_i$. Again, the exact effect on validator $\vldr$'s deposit depends on whether $\vldr$ voted successfully during epoch $i$ or not.
\begin{description}[itemindent=0cm]
\item[Case 2a: $\mathbf 1_{\vldr,i}=1$.] In this case, $\deposit_{\vldr,i+1}=\(1+C_i\) \deposit_{\vldr,i}=\(1+\frac12m_i\rho_i\)D_{\vldr,i}$, i.e., validator $\vldr$'s deposit increases by the collective reward factor. Since, $m_i>2/3$, and possibly even $m_i\approx 1$ under default conditions, we obtain the approximation 
\[\deposit_{\vldr,i+1}\approx\deposit_{\vldr,i}+\frac12\rho_iD_{\vldr,i}\]
In other words, a voting validator $\vldr$'s deposit increases by (approximately) half the individual reward factor which in turn is decreasing in the total staked deposit (more validators means each validator earns less). 
\item[Case 2b: $\mathbf 1_{\vldr,i}=0$.] In this case, \[\deposit_{\vldr,i+1}=\frac{1+C_i}{1+\rho_i}\deposit_{\vldr,i}=\frac{1+m_i\rho_i/2}{1+\rho_i}D_{\vldr,i}.\] Since more than $2/3$ of validators are voting successfully (epochs are being finalised), we have that $m_i\rho_i/2\in [\rho_i/3,\rho_i/2]$. Hence, validator $\vldr$'s deposit decreases. The intensity of the decrease depends on the order of magnitude of $\rho_i$ which is typically quite small.
\end{description}
\end{description}
The above analysis is summarised in \Cref{tab:cases}. 
\begin{table}[!htb]
\centering
\setlength{\tabcolsep}{12pt}
\renewcommand{\arraystretch}{1.2}
\begin{tabular}{@{}cclll@{}}
\midrule
\b Epoch Finalization & \b Validator $\vldr$'s vote &\b Validator $\vldr$'s new deposit &\b Parameters \\ 
\midrule
\textbf{ESF$_i>2$} &  $\mathbf 1_{\vldr,i}=1$ & $D_{\vldr,i+1}=D_{\vldr,i}$& $\rho_i=\gamma\cdot D_i^{-p}+\beta\(\text{ESF}_i-2\)$\\
& $\mathbf 1_{\vldr,i}=0$ & $D_{\vldr,i+1}=\frac1{1+\rho_i}D_{\vldr,i}$\\
\midrule
\textbf{ESF$_i=2$} &  $\mathbf 1_{\vldr,i}=1$ & $D_{\vldr,i+1}=\(1+m_i\rho_i/2\)D_{\vldr,i}$ & $\rho_i=\gamma\cdot D_i^{-p}, m_i\in[2/3,1]$\\
& $\mathbf 1_{\vldr,i}=0$ & $D_{\vldr,i+1}=\frac{1+m_i\rho_i/2}{1+\rho_i}D_{\vldr,i}$\\
\bottomrule
\end{tabular}\vspace{0.35cm}
\caption{Updates in validators deposits.}
\label{tab:cases}
\end{table}

\noindent In the Casper contract, the reward factors are updated by calling the function {\tt initialize\_epoch} once per epoch (the contract processes only the first call in each epoch).

\subsubsection*{Slashing Conditions}\label{sub:slashing}
Finally, to account for conflicting votes, Casper introduces the following \emph{commandments} or \emph{slashing conditions}, \cite{Bu1811}: a validator $\vldr$ who publishes two distinct votes
\begin{equation*}
\langle \vldr, \checkpointb_1, \text{\h(\checkpointb\textsubscript{1})}, \text{\h(\checkpointa\textsubscript{1})}, \signature_1 \rangle \hspace{0.3cm} \text{and} \hspace{0.3cm} \langle \vldr, \checkpointb_2,  \text{\h(\checkpointb\textsubscript{2})}, \text{\h(\checkpointa\textsubscript{2})}, \signature_2 \rangle,
\end{equation*}
violates Casper's slashing conditions, if
\begin{enumerate}[leftmargin=*,align=left]
\item[\textbf{I.}] \h(\checkpointb\textsubscript{1}) $=$ \h(\checkpointb\textsubscript{2}): i.e., they are for same target height,
\item[\textbf{II.}] \h(\checkpointa\textsubscript{1}) $<$ \h(\checkpointa\textsubscript{2}) $<$ \h(\checkpointb\textsubscript{2}) $<$ \h(\checkpointb\textsubscript{1}): i.e., one vote is within the \emph{span} of the second vote.
\end{enumerate}
Any validator is able to call the {\tt slash} function with the data for two potentially offending vote messages as its arguments. If the {\tt slash} call is found to be valid, then the offending validator's deposit is partially or entirely taken away (\emph{slashed}) and the sender receives 4\% of the validator's deposit as a `finder's fee'. Invalid calls to {\tt slash} are ignored in the same manner as invalid votes. Calls to {\tt slash} can be valid on any chain, even those chains that do not include one or more of the offending {\tt vote} calls in any of its blocks, because the function arguments and signatures by themselves are sufficient evidence for misbehaviour. In the Casper contract, the degree to which the validator's deposit is shrunk depends to the total amount slashed `recently' across the protocol -- the reasoning is to punish more harshly when there is a higher risk that two conflicting checkpoints with the same height will be finalised.

%, see \cite{Bu1811} or \cite{moindrot2017proof} for more details.

%\subsubsection*{Implementation of the Reward Scheme}
%
%
%
%\begin{itemize}
%\item  In the scheme, voters are rewarded (and non-voters penalised) implicitly, via two variables that are updated after every checkpoint: 
%\begin{itemize}[leftmargin=*]
%\item the \emph{deposit scale factor}, denoted by $\scalefactor$, and
%\item the \emph{reward factor}, denoted by $\rewardfactor$.
%\end{itemize}
%In the first block of each epoch, $\scalefactor$ is decreased, but a validators who votes receives a reward $\rewardfactor$ that (more than) compensates for the loss in terms of their \emph{scaled} deposit.
%\item Since $\scalefactor$ and $\rewardfactor$ vary from epoch to epoch, they can be described using sequences $\{\scalefactor_i\}_{i=0,1,\ldots}$ and $\{\rewardfactor\}_{i=0,1,\ldots}$ where $\scalefactor_i$ and $\rewardfactor_i$ are their values in epoch $i$, i.e., during blocks in $\(i\epochlength, (i+1) \epochlength\right]$. Their exact values depend on the voting behaviour of the validators. 
%\item A vote is typically sent when some fraction of the epoch has passed: in Pyethapp (one of the Ethereum clients that was used to test the Casper contract on Ethereum's test net) it is done after $\frac{1}{4}$ of the epoch \cite{floersch2017pyethapp}. 
%\end{itemize}

\section{Analysis}\label{sec:analysis}

In this section we investigate how Casper's incentive mechanism affects the fundamental properties of the overall consensus protocol, i.e., on the underlying block proposal mechanism combined with Casper FFG's checkpointing scheme. Before we begin the analysis, we discuss our threat and execution model in Section~\ref{sec:model}. We proceed by analysing the fundamental properties of \textit{liveness} and \textit{safety} in Section~\ref{sec:liveness}~and~\ref{sec:safety}, respectively. We focus on the following types of liveness and safety faults for checkpointing protocols.\footnote{See also \cite{Gi17,bentov2016snow,kiayias2017ouroboros} and \cite{Ro18} for varying definitions of these concepts.}  \begin{description}[leftmargin=0cm]
\item[Liveness faults:] checkpoints not getting finalised during one or more consecutive epochs. %with finalisation not being able to resume. 
%To account for this problem, Casper's rewards (penalties) are increasing (decreasing) in the proportion of validators that are voting. 
%Hence, the design of a mechanism that incentivises participation to increase security but also adequately penalises liveness faults is therefore difficult.
%given only a transcript of sent messages that contain a record of a user $S$ sending a message and the absence of an expected message from user $A$, this could arise because $A$ was not speaking, or because $B$ was not listening, and \emph{there is no way to tell the two apart}. 
%Since, in a liveness fault, it cannot be unambiguously determined who was at fault, and this  After all, if the penalties cause innocent validators to not feel comfortable participating, then this itself reduces security. This will be the subject of the next section. To account for the speaker/listener fault equivalence mentioned in the previous section, these rewards (penalties) are increasing (decreasing) in the proportion of validators that are voting. 
\item[Safety faults:] two or more conflicting checkpoints being finalised in the same or different epochs. (After all, this would either lead to a permanent fork, or to at least one node having to overturn a finalised checkpoint through a manual reset.)
\end{description}
Finally, we study \textit{incentive compatibility} (i.e., whether validators are incentivised to follow the protocol) in Section~\ref{sec:incentivecomp}.

%\todo{Are  PBFT protocols actually live under our definition, or can they freeze if $\alpha > \frac{1}{3}$?}

%In contrast to the definition of liveness in which $\alpha$ may be arbitrarily small, here $\alpha>0.5$ which corresponds to the minimum assumption of an honest majority. Given Casper's fork-choice rule and consensus mechanism, safety amounts to no two conflicting checkpoints being finalised. \todo{hmm, is this true? A minority attacker can create a minority fork with finalised checkpoints after a few months.}

\subsection{Execution Model}\label{sec:model}

\subsubsection*{Node Types}

We assume that the protocol consists of two main types of nodes: \emph{block proposers}, who propose blocks on the underlying chain, and \emph{validators} who vote for checkpoint blocks.
We use the term \emph{$\alpha$-strong} for groups of nodes who control \emph{at least} a fraction $\alpha \in (0, 1]$ of the resources necessary for their role in the protocol (e.g., stake for validators, or mining power for PoW block proposers). When we state that a certain group of nodes is $\alpha$-strong during a certain period, e.g., until finalisation resumes, then we assume that their strength at any point in time is at least greater than $\alpha$ even if it varies across the period. For example, if the group of online nodes is $60\%$-strong at the start of a period, but $20\%$ go offline in the middle of the same period, then we say that the online nodes are at most $40\%$-strong throughout the period.
Since the Casper contract does not add new validators to the active validator set when checkpoints are not being finalised, and most of the proofs in this section consider periods without finalisation, changes to the strength of validator groups are in practice limited to network partitions, latency, validators going online/offline, and the effects on their deposits of rewards and penalties.

The theorems in Sections~\ref{sec:liveness}~and~\ref{sec:safety} reason about the liveness and safety of the network as seen from the perspective of a single user, and only hold if the fraction of the consensus nodes who \emph{support} that user's branch is strong enough. In particular, the supporting nodes are those that are honest (i.e., protocol-following), online (i.e., non-faulty, and not being censored or eclipsed), and aware of the same input as the user (i.e., on the same branch during a network partition). We denote by $\alpha$ the strength of the validators who support the chain seen by the user.  Similarly, we denote by $\mu$ the strength of the block proposers who support the chain of the user.

%To study liveness and safety, we assume an adversary that controls at most $49\%$ of Casper's resources (stake) for an infinite period of time. We note that Casper also has a defensive mechanism against majority $51\%$ attacks, namely the \emph{minority fork} \cite{ryan7}. Here, we will focus on the standard threat model that assumes an honest majority. We assume a partially synchronous network, i.e., messages that do not arrive within the timespan of the intended epoch, are ignored. As we discuss further in \Cref{sub:limitations}, we do not account for collusions between miners and validators.

\subsubsection*{Network Synchrony}

In principle, we assume that the network is partially synchronous \cite{dwork1988consensus}, which means that there is a finite upper bound on the time before a message from any node $A$ reaches any other node $B$, but that we do not know the exact value of this bound. This has consequences for our failure model: e.g., we do not require that a vote message always reaches a block proposer before the end of the epoch, which would lead to a validator suffering a liveness penalty. However, when we say that the supporting validators are $\alpha$-strong, we do require that the \emph{proportion} of the nodes whose messages make into the blockchain in time does not go below $\alpha$. In other words, the fraction $1-\alpha$ of non-supporting validators during an epoch includes those whose messages did not make it onto the blockchain before the end of the epoch due to network latency.

In the benchmark parameter setting, the duration of an epoch is $700$ (or $50$ times $14$) seconds. Each validator can individually decide when to cast a vote - this is a trade-off between 1) voting too early and therefore risking that her view of the blockchain is out-of-date  2) voting too late and therefore risking that her vote does not appear on the blockchain in time. For example, a validator who casts her vote after a quarter of the epoch has 175 seconds since the start of the epoch to decide on a block and 525 to get the message on the blockchain. In practice, we assume that these periods are long enough to avoid network latency having a significant impact on the strength of the supporting nodes.

During a network partitions, some nodes will not be able to reach every other node within a bounded time (that is, before the partition is resolved). This case is treated separately in the proofs in this section -- as mentioned before, nodes that cannot communicate with a user are not counted as supporting that user's chain.

\subsubsection*{Interactions Between Block Proposers and Validators}

In \qt{hybrid} Casper FFG, there is some interdependence between the validators and the block proposers in the sense that the protocol messages sent by the former need to be included in blocks proposed by the latter. For example, malicious block proposers could 1) refuse to include vote messages of blacklisted validators in their blocks, 2) refuse to extend blocks that include those messages, or 3) attempt to intimidate other block proposers into enforcing the blacklist through feather forking \cite{miller2013feather}. However, the probability of successfully performing the former two attacks is low if the strength of the malicious block proposers is noticeably below $0.5$, and if the vote messages are sent early within the epoch. For example, for attack 1 it holds that if the malicious nodes are $(1-\mu)$-strong, and if virtually all block proposers have seen a validator's vote after the $k$th block within an epoch, then the probability that not a single block includes the message before the end of the epoch equals $(1-\mu)^{l-k}$. For $1-\mu = \frac{1}{3}$ and $k=20$, this probability is approximately equal to $4.9 \cdot 10^{-15}$. 
For attack 2 to succeed, the malicious nodes need to build a chain of length $l-k$ before the honest nodes do, so that they can override the honest nodes' chain (which may include the blacklisted validator's vote) at the end of the epoch. This probability is similarly low.\footnote{It can be expressed in terms of the probability of reaching some set of absorbing states in a discrete-time Markov chain. Numerical algorithms exist to compute such probabilities, e.g., those implemented in the model checking tool PRISM \cite{kwiatkowska2002prism}, but a discussion of this would be outside the scope of this paper. }
Feather forking is theoretically feasible even when the malicious nodes are less strong, but requires that the honest nodes comply despite invariably hurting the confidence that users have in the platform by doing so, and therefore reducing their own profits. Finally, a malicious coalition of block proposers could try to censor slash messages. Since these messages are not epoch-dependent, the censorship would need to be maintained indefinitely, which is impossible in practice if the malicious nodes are not at least $0.5$-strong.

The role of potentially malicious block proposers in the liveness proof will be discussed more detail in Section~\ref{sec:livenessdiscussion}, and we will make our assumptions regarding the fraction of supporting block proposers explicit for Theorem~\ref{prop:fork}. We note that in the future, Ethereum intends to adopt a fully PoS setting, in which case the supporting block proposer and validator groups always have equal strength.

As a final note, we would like to emphasise that a coalition of validators can also overrule a majority of the block proposers by voting for checkpoints that are different from those prescribed by the longest-chain rule. If the validators are \mbox{$\frac{2}{3}$-strong}, then they have complete freedom to justify and finalise their preferred checkpoints. This would enable a supermajority of activist validators to act against attacks by a majority of block proposers (e.g., long-term censorship of slash messages). Furthermore, any other group of validators has the potential to fork away from the main chain, and to eventually become \mbox{$\frac{2}{3}$-strong} on their own chain through the liveness recovery mechanism. This procedure is called a \textit{minority fork}, but detailed analysis of this mechanism is considered out of scope for this paper.

%\subsubsection*{The Underlying Chain}
%
%We assume that the underlying follows some longest-chain protocol that satisfies its own liveness and safety properties. 
%
%
%
%Randomness stems from block creation times in the underlying PoW chain. Block creation times are typically assumed to be exponentially distributed.
%
%\subsubsection*{Canonical chain}
%
%Explain what we mean by the canonical chain.
%
%What if the validators and block proposers are split 

\subsection{Liveness}\label{sec:liveness}

%Let \mbox{$T_\alpha \in\(0,\infty\)\cup\{\infty\}$} denote the random time until the finalisation of the next checkpoint by a set of $\alpha$-strong nodes.

In this section, we focus on checkpoint liveness as defined below.

\begin{definition}[Checkpoint liveness]\label{def:live}
We say that a checkpointing protocol $\Pi$ is \emph{$\alpha$-live} for some $\alpha \in \(0,1\right]$, if any $\alpha$-strong group of validators who support a chain are able to finalise checkpoints on that chain after a finite number of epochs with probability $1$. 
\end{definition}

During any period $E \subset \mathbb{N}$ of consecutive epochs in which more than $1/3$ of the validators (in terms of stake) are not voting correctly, checkpoints cannot be finalised. In the reward mechanism, this translates to an increase in the epochs since finalisation, ESF$_{i}$ which in turn affects both the individual $\rewardfactor_i$ \eqref{eq:rewardfactor} and the collective $\collectivereward_i$ \eqref{eq:collective} reward factors, for $i\in E$. This is captured by \Cref{lem:technical}. Without loss of generality, we assume that a fault involving ($1-\alpha$)-strong validators not following the protocol is initiated at epoch $0$. All relevant notation is summarised in \Cref{tab:base_losses}.

\begin{table}[!htb]
\centering
	\begin{tabular}{@{}lcl@{}}
	\midrule
	\textbf{Notation} && \textbf{Description} \\ \midrule
	%$H$ & & the Ethash cryptographic hash function \\
	$l$ & \; & length (in terms of blocks) of a Casper FFG epoch \\
	$V$ & & current active validator set \\
	$\deposit_{\vldr,i}$  & & deposit of validator $\vldr \in \mathbb{N}$ at the beginning of epoch $i$ \\
	$\deposit_{i}$  & & total deposit at the beginning of epoch $i$ \\
	$m_i\in[0,1]$ & & weighted fraction of correct votes in epoch $i$\\
	$\rewardfactor_i\ge0$ & & individual reward factor in epoch $i$\\
	$\collectivereward_i\ge0$ & & collective reward factor in epoch $i$\\
	$\baseinterest>0$ & & base interest factor $>0$ \\
	$\basepenalty>0$ & & base penalty factor \\
	 $\depsizefactor>0$ & & total deposit dependence factor (typically in $(0,1]$)\\
	 $\alpha\in (0,1]$ & & stake fraction of supporting validators\\
	 $\mu\in (0,1]$ & & stake fraction of supporting block proposers\\\midrule
	\end{tabular}\vskip0.1cm
\caption{List of symbols.}
\label{tab:base_losses}
\end{table}

\subsubsection{Delaying Finalisation: Liveness Recovery from Offline Users}

We begin with a discussion of the scenario where checkpoints are not being finalised due to a sufficiently large number of nodes having gone permanently offline. For ease of notation, we consider the case where a single validator $\vldr$ represents the supporting nodes, and another validator $\nvldr$ the non-supporting nodes.

\begin{lemma}\label{lem:technical}
Let $\deposit_0$ denote the initial stake in epoch $0$ and let $\vldr, \nvldr$ be $\alpha$ and $\(1-\alpha\)$-strong validators, respectively, with $\alpha\in \(0,2/3\)$. Assume that from epoch $0$ onwards, only the $\alpha$-strong validator $\vldr$ continues to vote correctly. Then, for the consecutive epochs $i\ge1$ that are not being finalised
\begin{enumerate}[label=(\roman*)]
\item the total deposits $\deposit_i$ are decreasing in $i$,
\item $\deposit_i$ is given by $\textstyle\deposit_i=\alpha \deposit_0+\(1-\alpha\)\deposit_0\prod_{j=0}^{i-1}\(1+\rewardfactor_j\)^{-1}$,
\item the individual reward factors $\rewardfactor_i$ are unboundedly increasing in $i$.
\end{enumerate}
\end{lemma}

\begin{proof} By definition \eqref{eq:collective}, as long as epochs $i\ge1$ are not being finalised, $C_i=0$ . Hence, by \eqref{eq:interest} \[\deposit_{\vldr,i}=\(1+0\)\frac{1+1\cdot \rewardfactor_{i-1}}{1+\rewardfactor_{i-1}} \deposit_{\vldr,i-1}=\deposit_{\vldr,i-1},\] and similarly $\deposit_{\nvldr,i}=\(1+\rewardfactor_{i-1}\)^{-1}\deposit_{\nvldr,i-1}<\deposit_{\nvldr,i-1}$. The last inequality holds because $\rewardfactor_i>0$ for all $i\ge1$ by definition. Since $\deposit_{i}=\deposit_{\vldr,i}+\deposit_{\nvldr,i}$ for all $i\ge1$, this implies $\deposit_i=\deposit_{\vldr,i}+\(1+\rewardfactor_{i-1}\)^{-1}\deposit_{\nvldr,i-1}<\deposit_{i-1}$, which shows (i).
%that during the period of non-finalisation of checkpoints, the total deposits $\deposit_i, i=1,2,\dots$ are decreasing.
The expression in (ii) is obtained by repeated application of the previous recursive equalities, $D_{\vldr,i}=D_{\vldr,i-1}$ and $D_{\nvldr,i}=\(1+\rewardfactor_{i-1}\)D_{\nvldr,i-1}$, and the assumption $D_{\vldr,0}=\alpha D_0$. Finally, to prove (iii), observe that ESF$_{i}-2=i$, since the last finalised epoch is \qt{$-2$} and hence\footnote{Since we assumed that the attack started in epoch $0$ to avoid unnecessary notation, $-2$ should be interpreted acccordingly, i.e. it is a convention to denote the last finalised epoch, $2$ epochs prior to the start of the attack.}, by \eqref{eq:rewardfactor} \[\rewardfactor_i=\gamma\deposit_i^{-p}+\beta\(\text{ESF}_i-2\)=\gamma\deposit_i^{-p}+\beta i.\] Since $\deposit_i$ is decreasing by (i) and $p>0$, $\rewardfactor_i$ is increasing to $\infty$ which concludes the proof.
\end{proof}

Intuitively, \Cref{lem:technical} states that as epochs are not getting finalised, Casper's implemented incentives mechanism retains the deposits of correctly and consistently voting validators constant -- i.e., validators who vote in each epoch without violating any slashing conditions -- and decreases the deposits of non-voting validators. Thus, the voting validators' stake as a \emph{share} of the total deposited stake increases. As the ESF increases, this process continues until the voting validators will account for more than $2/3$ of the total stake on the chain that they are voting on. After that, the ESF resets to $2$ and the updates are much less aggressive. However, the overall period that is required for this to happen depends on the initial proportion of their stake and the voting behavior of the adversary (or the accidentally off-line nodes). This is the subject of Section~\ref{sub:worst}.

\subsubsection{Delaying Finalisation: Liveness in the Worst-Case Scenario}\label{sub:worst}
The assumption in \Cref{lem:technical} that the $\(1-\alpha\)$-strong nodes never vote again (after the initiation of the attack), does not correspond to the worst-case scenario for the protocol. If the \qt{non-voting} validators vote irregularly, then this slows down the reduction process of their deposits. In particular, if all $\(1-\alpha\)$-strong malicious validators can coordinate to reach a total vote of \textit{just} below $2/3$ -- the finalisation threshold -- in each epoch, then they thwart finalisation while minimizing the amount of penalised deposits (and equivalently maximize the time that the protocol is unable to finalise checkpoints). This case is treated in \Cref{lem:worst} in which we account for the following types of validators:
\begin{enumerate}[label=(\roman*), noitemsep]
\item Consistently supporting validators. Let $\alpha_i,i\ge0$ denote the deposit of the of $\alpha$-strong validators who vote consistently during all consecutive epochs $i\ge0$ for which finalisation does not happen. The size of the deposits of these validators in the initial period equals $\alpha_0\in\(0,2/3\)$.
\item Irregularly voting validators. To account for the worst-case scenario, we assume that the remaining $1-\alpha_0$ validators (in terms of their initial stake in epoch $0$), can coordinate under a single entity and decide their behavior in order to delay finalisation as much as possible. To model this situation, we assume that in each epoch $i\ge0$, a fraction $\delta_i$ of these validators are voting, where $\delta_i$ is the maximum possible number in $\(0,1\right]$ -- if any -- so that $\alpha_i+\delta_i\(1-\alpha_i\)\le2/3$, i.e., so that finalisation is marginally not possible. This policy maximizes the \qt{life} of the adversarial stake while preventing finalisation of checkpoints.
\end{enumerate}
Our goal is to prove that the sequence of periods for which such a $\delta_i\in\(0,1\right]$ exists is bounded, i.e., that for any initial honest stake $\alpha_0\in \(0,2/3\)$ finalisation of checkpoints will ultimately resume (for $\alpha>2/3$, the claim is trivial). This is achieved in \Cref{lem:worst}. To model the worst-case scenario, we assume that all validators not belonging to the group of $\alpha_0$-strong validators are malicious. In reality, some of them may only accidentally go off-line for some period and resume proper voting after that or vote irregularly during the period of non-finalisation. Also, note that during periods of non-finalisation, new validators cannot enter the protocol, i.e., deposits of new stake are not processed and reflected in the active validator set. This is critical to ensure safety of checkpoints in the PoS environment. 

\begin{theorem}\label{lem:worst}
Let $\deposit_{\vldr,0},\deposit_{\nvldr,0}$ denote the initial stake in epoch $0$ of the $\alpha_0$ and $\(1-\alpha_0\)$-strong validators $\vldr$ and $\nvldr$, respectively, with $\alpha_0\in \(0,2/3\right]$. Assume that from epoch $0$ onwards, only the $\alpha$-strong validator $\vldr$ continues to vote consistently and correctly. Also assume that in each epoch $i>0$, a fraction $\delta_i\in (0,1]$ of the $\(1-\alpha\)$-strong validators is also voting correctly, so that for each epoch $i\ge0$, $\delta_i$ is the solution to the optimization problem 
\begin{alignat*}{3}
\max & \quad \delta_i  \\
\text{subject to:} &\quad \alpha_i+\delta_i\(1-\alpha_i\)\le 2/3 \tag{WC}\\
&\quad \delta_i  \ge 0.
\end{alignat*} 
Then, the number $T\ge1$ of consecutive epochs that are not being finalised is finite.
\end{theorem}

\begin{proof}
Since, any feasible solution of (WC) (where WC stands for the worst-case policy) must satisfy the constraints of (WC), we have that the optimal solution $\delta_i$, whenever it exists, will be given by 
\begin{equation}\label{eq:opt}
\delta_i=\frac{2/3-\alpha_i}{1-\alpha_i}.
\end{equation}
Hence, (WC) is feasible if $\alpha_i\le 2/3$. Accordingly, to prove the claim, we need to show that there exists a $T\ge1$ such that $\alpha_T>2/3$. According to \Cref{tab:cases}, the deposit-shares at the start of epoch $t+1$ of the initially $\alpha_0$-strong honest validators and the $\(1-\alpha_0\)$-strong validators who follow the (WC) plan will be
\begin{align*}
\deposit_{\vldr,i+1}&=\deposit_{\vldr,i}, \quad \forall i\le T\\
\deposit_{\nvldr,i+1}&=\delta_i\deposit_{\nvldr,i}+\(1-\delta_i\)\(1+\rewardfactor_i\)^{-1}\deposit_{\nvldr,i}\overset{\eqref{eq:opt}}=\frac{\deposit_{\nvldr,i}}{1-\alpha_i}\left[\frac23-\alpha_i+\frac13\(1+\rewardfactor_i\)^{-1}\right], \quad \forall i\le T.
\end{align*}
Hence, 
\begin{align*}
\alpha_{i+1}&=\frac{\deposit_{\vldr,i+1}}{\deposit_{\vldr,i+1}+\deposit_{\nvldr,i+1}}=\frac{\deposit_{\vldr,0}}{\deposit_{\vldr,0}+\frac{\deposit_{\nvldr,i}}{1-\alpha_i}\left[\frac23-\alpha_i+\frac13\(1+\rewardfactor_i\)^{-1}\right]}\\&\overset{(*)}\ge
\frac{\deposit_{\vldr,0}}{\deposit_{\vldr,0}+\frac{\deposit_{\nvldr,0}}{1-\alpha_i}\left[\frac23-\alpha_i+\frac13\(1+\rewardfactor_i\)^{-1}\right]}=\frac{1}{1+\frac{1-\alpha_0}{1-\alpha_i}\left[\frac23-\alpha_i+\frac13\(1+\rewardfactor_i\)^{-1}\right]},\\[-0.2cm]
\end{align*}
where for the inequality $(*)$, we used that $\deposit_{\nvldr,i}<\deposit_{\nvldr,0}$ by \Cref{lem:technical}-(i). Now, since the limit $\alpha$ of $\alpha_i$ as $i\to \infty$ exists -- the sequence $\alpha_i$ is increasing by \Cref{lem:technical}-(i) and (iii) and bounded by $1$ -- and since $\(1+\rewardfactor_i\)^{-1}$ vanishes as $i\to\infty$, we can take limits in the last inequality to obtain 
\begin{align}\label{eq:inequality}
\alpha\ge\frac{1-\alpha}{1-\alpha+\(1-\alpha_0\)\(\frac23-\alpha\)}.
\end{align}
Hence, the possible limits of the sequence $\alpha_i$ is the locus of all points $\alpha\in [0,1]$ that satisfy inequality \eqref{eq:inequality} for any given value of $\(1-\alpha_0\)\in[1/3,1)$ (adversarial stake at the outset of the attack). To obtain a contradiction, assume that $\alpha\le 2/3$. Then, $\(1-\alpha_0\)\(2/3-\alpha\)\le0$ for any initial honest stake $\alpha_0\in(0,2/3]$ and \eqref{eq:inequality} becomes
\[\alpha\ge\frac{1-\alpha}{1-\alpha+0}=1
\] 
which is a contradiction, since we have assumed that $\alpha\le 2/3$. Hence, in the limit it holds that $\alpha>2/3$, which implies that there exists $T\in \mathbb N$, i.e., finite, such that the stake $\alpha_t$ of voting validators at each epoch $t\ge T$ exceeds $2/3$ and hence finalisation of checkpoints resumes. This concludes the proof.
\end{proof}

\subsubsection{Checkpoint Liveness}

Even in the case described in \Cref{lem:worst}, the adjustment of the deposits of non-voting validators will continue until the (consistently) voting validators will acount for more than $2/3$ of the total stake on the chain that they are voting on. From this point on, they will resume finalisation of checkpoints and hence, the ESF will reset to $2$, hence making changes in deposits much slower. This is the main argument behind Casper's liveness which is formalised in the \Cref{thm:liveness} and illustrated in \Cref{fig:liveness}.  
\begin{theorem}\label{thm:liveness}
The Casper contract is $\alpha$-live for any $\alpha\in\(0,1\right]$.
\end{theorem}
\begin{proof}
If the supporting validators control more than $\frac{2}{3}$ of the stake, then finalisation and hence, liveness are immediate. To treat the remaining case, assume that voting validators control $\alpha<\frac{2}{3}$ of the stake at epoch $0$. Without loss of generality\footnote{Possibly after renumbering epochs so that the (WC) policy mentioned in the worst case above, does not apply. In any case, the intention is to provide an intuitive proof about Casper's checkpoint liveness devoid of the technicalities of the worst-case scenario that was treated in \Cref{sub:worst}. All arguments carry over to that case as well but with different bounds on the time before  finalisation is resumed.}, also assume that after epoch $0$ the remaining $\(1-\alpha\)$-strong validators stop voting. In this case, finalisation will resume after epoch $k\ge1$, if $\deposit_{\vldr,k}\ge\frac{2}{3}\deposit_k$ or equivalently if $\deposit_k\le\frac{3}{2}\alpha\deposit_0$, since, by the proof of \Cref{lem:technical}, $\deposit_{\vldr,k}=\deposit_{\vldr,0}$. Hence, by \Cref{lem:technical}(ii)
\[\textstyle\frac{3}{2}\alpha\deposit_0\ge\alpha \deposit_0+\(1-\alpha\)\deposit_0\prod_{i=0}^{k-1}\(1+\rewardfactor_i\)^{-1},\]
which is equivalent to $\prod_{i=0}^{k-1}\(1+\rewardfactor_i\)\ge 2\(1-\alpha\)/\alpha$. This implies that finalisation will resume after epoch $k_\alpha$, where $k_\alpha$ is the solution to the following minimisation problem 
\[k_\alpha:=\textstyle\min {\left\{k\in \mathbb N : \prod_{i=0}^{k-1}\(1+\rewardfactor_i\)\ge 2\(1-\alpha\)/\alpha\right\}}.\]
Hence, it remains to show that the above problem has a finite solution $k_\alpha\in\mathbb N$ for every $\alpha\in\(0,\frac{2}{3}\)$. Since \mbox{$\rewardfactor_i=\gamma \deposit_i^{-p}+\beta i\ge \beta i$} by the proof of \Cref{lem:technical} (iii), the standard inequality $\prod_{i=0}^{k-1}\(1+\rewardfactor_i\)\ge 1+\sum_{i=0}^{k-1}\rewardfactor_i$, implies that it suffices to find a $k$ such that $\textstyle\sum_{i=0}^{k-1}\beta i=\beta\frac{k\(k-1\)}{2}\ge 2\(1-\alpha\)/\alpha$. Since $k^2$ is unbounded and $\beta>0$ is constant, such a $k$ exists for every $\alpha\in\(0,2/3\)$.
\end{proof}

\begin{figure*}[!htb]
\centering
\includegraphics{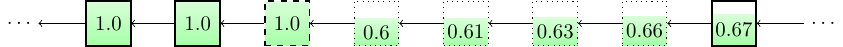}
\caption{Casper's liveness property, cf. \Cref{thm:liveness}: if some validators stop voting, the fraction of stake controlled by properly voting validators is adjusted over time to account for more than $2/3$ of the total stake. Finalisation resumes after some period of checkpoints that could not be finalised. In Figures~\ref{fig:liveness}~and~\ref{fig:safety}, \emph{solid} frames represent finalised epochs, \emph{dashed} frames justified epochs, and \emph{dotted} frames epochs that are neither justified nor finalised. The numbers inside the frames represent the stake fractions of the voting validators (in an exaggerated parameter setting).}
\label{fig:liveness}
\end{figure*}

Given the benchmark parametrisation of the contract, cf. \Cref{sub:scheme,sub:choice}, the number of epochs needed to resume finalisation can be computed numerically: let $\phi(\alpha)$ be the minimum number of epochs required before finalisation is resumed if an $\alpha$-strong group of of validators permanently stops voting. The evolution of $\phi(\alpha)$ is illustrated in \Cref{fig:finalization} for different values of $\alpha$.\footnote{\Cref{lem:technical} and \Cref{thm:liveness} assume that faulty validators stop voting \emph{completely}. More elaborate voting/non-voting strategies which can delay finalisation beyond the rate given by $\phi(\alpha)$ are discussed in Section~\ref{lem:worst}.}. We emphasise the following values which we will use in the study of Casper's safety properties: for $\alpha = 0.67$, $0.51$, and $0.49$, the number of epochs needed for $\alpha$-strong validators to resume finalisation is  3733, 2698, and 2546 respectively. 
%Since we have proven a lower bound on the number of epochs needed to finalise, our results generaise readily to \emph{any} block proposal that is live in the sense that .
\begin{figure}[!htb]
\centering
\includegraphics[trim={0 0.2cm 0 0}]{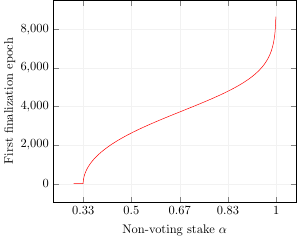}
\caption{Epoch of the first finalisation on a chain with non-voting validators controlling $1-\vdsize$ of the (initial) total stake.}
\label{fig:finalization}
\end{figure}

\subsubsection{Discussion of the Liveness Proof}\label{sec:livenessdiscussion}

We conclude the section on liveness with a discussion about the implications of some of the assumptions of Theorem~\ref{thm:liveness}. First, the theorem states that liveness is recovered within a finite number of epochs, but the \textit{time duration} of those epochs depends on the underlying block proposal mechanism. For example, if an $(1-\mu)$-strong group of miners in Ethereum's current PoW chain also go offline, then the average time between blocks goes up to $14/\mu$ seconds. If \textit{all} miners have gone offline (i.e., $\mu = 0$), then the number of epochs before liveness recovery is still finite, but the time duration of each epoch is infinite. However, the PoW chain's permissionless nature allows anyone (including the validators) to join the protocol, so this situation would be unlikely to persist in practice.

As discussed in Section~\ref{sec:model}, the block proposers also have other methods to influence liveness, e.g., by censoring vote messages. Note that our definition of $\alpha$-strong support precludes more than an $(1-\alpha)$-strong group of validators being censored by the block proposers. Also, a coalition of validators can ignore the standard fork-choice rule to prioritise a chain that includes their votes. As such, the proof of checkpoint safety holds for all supporting block proposer strengths $\mu \in (0,1]$, assuming that the underlying block proposal mechanism is still live (which is the case for Ethereum's PoW chain).

Finally, although we have proven the liveness of Casper's checkpointing mechanism, we have so far not discussed its impact on transaction liveness, i.e., the ability of honest nodes to append transactions to the ledger. If a coalition of malicious validators is at least $\frac{2}{3}$-strong, then it is free to perform censorship by refusing to vote for checkpoints that include a blacklisted transaction. Furthermore, as we will see in the next section, if the malicious-nodes are $1-\alpha$-strong with $\alpha \in (\frac{1}{3}, \frac{1}{2})$, then they can vote for checkpoints on a censored chain, and be able to finalise checkpoints before the $\alpha$-strong nodes that support the main chain (i.e., the chain that observes the protocol's fork-choice rule) do so. Hence, although the checkpointing protocol is by itself live for any $\alpha > 0$, it will only produce meaningful checkpoints for users if $\alpha > \frac{1}{2}$.\footnote{If $\alpha \leq \frac{1}{2}$, the system may still be recovered using a minority fork as mentioned in Section~\ref{sec:model}.}

\subsection{Safety}\label{sec:safety}

Having proved the liveness of our checkpointing protocol, we now move on to safety, which we define as follows:

\begin{definition}[Checkpoint safety]\label{def:safety}
We say that a checkpointing protocol $\Pi$ is \emph{safe} if the following holds: any protocol-following node who considers a checkpoint $i$ as final at some time point $\tau \geq 0$, will also consider $i$ as final at any time point $t > \tau$.
\end{definition}

To explore the trade-off between liveness and safety, we distinguish two scenarios: in the first, we assume a unified network in which clients have a view of all active chains, and in the second, a partitioned network, in which each client has a view of only a single chain.\par 
In the first scenario, we seek to prevent the \emph{nothing-at-stake} problem, which occurs when there is an incentive to finalise conflicting checkpoints on different chains during a fork. Casper's mechanism ensures that in the short term, different checkpoints cannot be finalised unless at least $1/3$ of validators violate one slashing condition. Intuitively, this relies on the fact that a checkpoint can be finalised if and only if a direct child of this checkpoint has been justified, see \Cref{sub:finalisation}. Hence, for two conflicting checkpoints to have been finalised, there must exist two pairs of consecutive justified checkpoints on two different (conflicting) chains. Taking into account that votes between different chains are identified by the protocol as invalid and are ignored, two cases may have occured:
\begin{itemize}[leftmargin=*]
\item two of the conflicting checkpoints have the same height: this directly violates slashing condition I.
\item all conflicting checkpoints have different heights: this implies that one pair of consecutive justified checkpoints must be included within the span of two justified checkpoints on the conflicting chain, which violates slashing condition II.
\end{itemize}
This is the statement of \Cref{prop:slash}, whose proof formalises the above intuition and is similar in nature to \cite[Theorem 1]{Bu18}.

\begin{theorem}\label{prop:slash}
%\todo{Assuming that the network is not partitioned: it seems like a weird assumption that may annoy the referee. Can you rephrase/obscure to more vague assumption?}
Assuming that the network is not partitioned, two conficting checkpoints cannot be finalised unless at least $\frac{1}{3}$ of the validators violate a slashing condition. 
\end{theorem}
\begin{proof} Let $B_1$ and $B_1'$ denote two conflicting finalised checkpoints with direct children $B_2$ and $B'_2$, respectively. Also, let $M$ denote the maximal -- in terms of block height -- common element in the chains of the two conflicting justified checkpoints, i.e.,
\[M:=\arg\max{\left\{h\(B\): B \in C\(B_2\)\cap C\(B'_2\)\right\}}\]
$M$ exists since at least $g\in C\(B_2\)\cap C\(B'_2\)$ by assumption. Then, there exist $m,n\ge2$ such that 
\begin{align*}
C\(B_2\)&=\(B_2,B_1,P\(B_1\),\dots,P^n\(B_1\), C\(M\)\),\\
C\(B'_2\)&=\(B'_2,B'_1,P\(B'_1\),\dots, P^m\(B'_1\), C\(M\)\)
\end{align*} 
with $P^m\(B'_1\)\neq P^n\(B_1\)$ by definition of $M$. If there exist blocks 
\[B\in \(B_2,B_1,P\(B_1\),\dots,P^n\(B_1\)\) \quad \text{and} \quad B'\in \(B'_2,B'_1,P\(B'_1\),\dots,P^m\(B'_1\)\)\] 
with $h\(B\)\neq h\(B'\)$, then this creates a violation of slashing condition I. In other words, given that all blocks in the chains up to $B_2$ and $B'_2$ are different after $C\(M\)$, they must be at different heights for slashing condition I to not be violated. To proceed, let's assume that condition I is not violated, i.e., that $h\(B\)\neq h\(B'\)$ for all $B\in \(B_2,B_1,P\(B_1\),\dots,P^n\(B_1\)\)$ and $B'\in \(B'_2,B'_1,P\(B'_1\),\dots,P^m\(B'_1\)\)$. However, since $B_2$ and $B'_2$ are both direct children of $B_1$ and $B'_1$ respectively, this means that $h\(B_2\)=l+h\(B_1\)$ and $h\(B'_2\)=l+h\(B'_1\)$ where $l$ denotes the epoch length. In turn, this implies -- after possibly renaming $B'_1$ to $B_1$ and $B'_2$ to $B_2$ and hence without loss of generality -- that 
\[h\(B'_1\)<h\(B'_2\)<h\(B_1\)<h\(B_2\).\]
Hence, since $B'_2$ and $B'_1$ have consecutive heights (in terms of checkpoints), there must be a $k \in \{1,2,\dots, n\}$ with 
\[h\(P^k\(B_1\)\)<h\(B'_1\)<h\(B'_2\)<h\(P^{k-1}\(B_1\)\)\]
which creates a violation of slashing condition II. 
\end{proof}
\begin{remark}[Theorem~\ref{prop:slash} during a network partition]
We have required for the proof of Theorem~\ref{prop:slash} that the network is not undergoing a partition. This assumption can be loosened somewhat, but the duration of the partition does have an impact on the minimum amount of stake that needs to be slashed to achieve a safety fault. The reason is that the deposit sizes on the different branches of the fork evolve as the partition continues. For example, if the honest nodes are $80\%$-strong but split evenly across two branches during a fork, and the malicious nodes are 
$20\%$-strong but able to vote on both branches, then the malicious voters will not be able to cause a safety fault immediately. Also, since the honest nodes cannot see that the malicious nodes are voting on both branches, no slash messages will be broadcast for the duration of the partition. However, in the long term, the liveness recovery mechanism will ensure that the voting validators on both branches achieve a two-thirds supermajority. Although the validators who voted on both branches can be slashed in retrospect, the total amount that gets slashed is only $20\%$ of the deposit size at the start of the partition. 

For a complete analysis, the minimum amount of stake required to cause a safety fault would need to be expressed as a function of the number of epochs since the network was partitioned, and the distribution of the stake across the different branches of the blockchain. We do note that the penalty mechanism, which increases in intensity as the ESF grows larger, does not have a major effect during short-term partitions. In the remainder of this section, we will discuss the effect of long-term network partitions without considering the additional effect of votes that violate slashing conditions, and leave analysis of the precise interplay between slashing and network partitions as future work.
\end{remark}
We now turn to the scenario of a long-term network partition. In this case, Casper's checkpointing mechanism provides an additional layer of security which improves over classic PoW protocols against \emph{long-range attacks}. This is the statement of \Cref{prop:fork} which we prove for Casper's current parametrisation, cf. \Cref{sub:scheme,sub:choice}.  To address this case in full generality, we need to take into account the way that the network is partitioned. For this, we distinguish between the validators and the block proposers who support the user's chain, which we assume to be $\alpha$-strong and $\mu$-strong, respectively, with $\alpha,\mu \in [0,1]$. For example, if $\alpha=0.6$ and $\mu=0.2$ then $60\%$ of the deposit-weighted validators are on the user's side of the partition, but only $20\%$ of the mining power. 
%
%
%have the following distinction
%\begin{description}[noitemsep,leftmargin=*]
%\item[Checkpointing mechanism.] This refers to the validators and the Proof of Stake (PoS) mechanism. We will denote with $\alpha$ the fraction of validators voting on the chain with the majority of validators, i.e., $0.51\le \alpha \le 1$. 
%\item [Block production mechanism.] This may refer to the miners of the underlying Proof of Work (PoW) chain  -- in case that Casper's checkpointing mechanism is overlayed on top of PoW blockchain -- or again to the validators' themselves -- in case of a complete PoS implementation. In any case, we will use $\mu\in[0,1]$ to denote the fraction of miners (block creators) in the chain of interest\footnote{In the second case, that of a full PoS implementation, $\mu$ coincides with $\alpha$.}.
%\end{description}
Given this notation, we will refer to a chain with $\alpha$ fraction of validators and $\mu$ fraction of block proposers as an $\(\alpha,\mu\)$-chain. In what follows, we define the \emph{canonical chain} as the chain on which checkpoints are finalized first in terms of time (so not in terms of blocks or epochs). As mentioned before, we assume that this is a genuine network partition, i.e., that no validators vote on both branches of the chain. We then have the following.
%The intuitive reasoning for this improvement is that the majority chain -- which is assumed to be controlled by honest nodes -- will be able to finalise checkpoints ahead of any other chain with overwhelming probability. Hence, given Casper's fork-choice rule, validators voting on this chain can be sure that this will also be the \emph{canonical chain} in the future and that finalised checkpoints will not be overturned.\par
%\footnote{Similar bounds can be derived for any choice of parameters in their admissible range, cf. \Cref{sub:choice} and any number $n>2$ of forks.}.
%the idea readily extends to any number $n\in\mathbb N$ of forks and any choice of parameters in their admissible range, see \Cref{sub:choice}.
\begin{theorem}\label{prop:fork}
Consider a $\(\alpha,\mu\)$-chain with $\alpha>0.5$ and $\mu\in[0,1]$. Then, 
\begin{description}[leftmargin=0cm,noitemsep]
\item[Case I.] If $2/3<\alpha\le1$, the $\(\alpha,\mu\)$-chain will finalise checkpoints after at most $3$ epochs. For a conflicting checkpoint to be finalised, the partition should last at least $\phi(2/3) = 3733$ epochs. In this case, the probability that a checkpoint on the $\(\alpha,\mu\)$-chain will be finalised first, is at least $P\(B<\mu\)$ with $B\sim \text{Beta}\(3,3733\)$, which for $\mu\ge0.004$ is larger than $0.9999$.  
\item[Case II.] If $0.51<\alpha\le 2/3$ and $0.5<\mu\le 1$, the $\(\alpha,\mu\)$-chain will finalise checkpoints first with overwhelming probability, after at most $\phi(0.51) = 2546$ epochs. For a conflicting checkpoint to be finalised, the partition should last at least $\phi(0.49) = 2698$ epochs.
\end{description}
\end{theorem}
\begin{proof} 
\begin{description}[leftmargin=0cm]
\item[Case I.] In the first case, i.e., if more than $2/3$ of validators are voting on the same chain, then they will be able to finalise checkpoints  without any need for the liveness recovery mechanism. After the start of the partition, the conflicting chains will need to wait for less than one epoch before reaching the first potentially conflicting checkpoint. After another epoch, this checkpoint will be justified, and after another epoch it will be finalised, resulting in a total of 3 epochs. The amount of time that it takes for a corresponding number of blocks to be proposed depends on the speed of the block proposal mechanism, which in turn depends on the fraction $\mu$ of miners that are supporting the $\(\alpha,\mu\)$-chain. For the likelihood of a conflicting checkpoint, in the worst-case scenario, all remaining $1-\mu$ miners will be coordinated on the same (competing) chain as the remaining $1-\alpha$ validators. Thus, using the finalization bounds of \Cref{fig:finalization}, we have that the (at most) $1/3$ of validators on the competing chain will need (at least) $3733$ epochs to exceed the $2/3$ threshold (in terms of stake) on that chain and finalise checkpoints. Hence, that chain will be able to finalise first, if $\mu$ is so small that $3$ epochs of blocks on the $\(\alpha,\mu\)$-chain will take more time to be created than $3733$ epochs of blocks on the $\(1-\alpha,1-\mu\)$ chain. Since individual block creation times on a chain are exponentially distributed with mean equal to the fraction of the total mining power that is active on that chain, the time of the $n-$th block creation in a chain with $\mu$ mining power has the Erlang distribution with parameters $n$ and $1/\mu$. This follows from the well-known fact that the sum of $n$ exponential distributions with the same parameter $1/\mu$ has the Gamma distribution with parameters $\(n,\mu\)$, which is called the Erlang distribution when $n$ is an integer, cf. \cite{Gu13}. Thus, if we denote with $X$ the blocks on the $\(\alpha,\mu\)$-chain and $Y$ the blocks on the $\(1-\alpha,1-\mu\)$ chain, the probability that $X$ reaches 3 epochs of blocks after $Y$ reaches $3733$ of blocks is equal to
\begin{align*}
P\(X>Y\)&=P\(X\mu>Y\mu\pm Y\(1-\mu\)\)=P\(X\mu+Y\(1-\mu\)>Y\)=P\(\frac{Y\(1-\mu\)}{X\mu+Y\(1-\mu\)}<1-\mu\)=F_{B'}\(1-\mu\)
\end{align*}
where ${B'}$ has a Beta distribution with parameters $n_1=3733$ and $n_2=3$. This result is based on the well-known fact from (textbook) probability that if $X\sim \text{Erlang}\(n_1=3,\mu\)$ and $Y\sim \text{Erlang}\(n_2=3733,1-\mu\)$, then (i) $X/\mu\sim \text{Erlang}\(n_1,1\)$ and $Y/\(1-\mu\)\sim \text{Erlang}\(n_2,1\)$ and (ii) $Y/(X+Y)\sim \text{Beta}\(n_2,n_1\)$. Finally, (using properties of the Beta distribution), we can simplify the above equation to obtain that 
\[P\(X<Y\)=F_B\(\mu\), \hspace{15pt}\text{with } B\sim\text{Beta}\(n_1=3,n_2=3733\)\]
which is precisely the statement of Case I. 
\item[Case II.]  In this case, validators voting in the $\(\alpha,\mu\)$-chain will have to wait for at most $2546$ epochs, cf. \Cref{fig:finalization}, for their stake, $\alpha$, to account for $2/3$ of the total stake in that chain. Assuming that all other $1-\alpha$ validators are voting on a $\(1-\alpha,1-\mu\)$, then they will need at least $2698$ epochs to resume finalisation, cf. \Cref{fig:finalization}. With a similar argument as in Case I, the probability that finalisation in the $1-\mu$ chain will precede finalisation in the $\mu$ chain is less or equal to the probability that an Erlang random variable with parameters $n_1=2698,\mu_1=0.49$ will be less or equal than a random variable $n_2=2546, \mu_2=0.51$, which is negligible (Python calculation: 0E-537). \qedhere
\end{description}
\end{proof}

The bounds in all the cases treated in \Cref{prop:fork}, refer to the worst case scenario -- in terms of the likelihood that a safety fault will occur -- in which the network is perfectly split into two competing chains. In the case that $\alpha>0.5$ and $\mu<0.5$, then checkpoints on an $\(1-\alpha,1-\mu\)$-chain may be finalised first. However, as explained in \Cref{sub:finality_casper}, these cases refer to extreme network conditions and can be expected to rarely occur in practice. In case that they indeed occur, then either the network will permanently fork or one part will need to incur considerable economic losses to reunite. In such occasions, the conflicts are expected to be resolved by an off-chain governance mechanism.\par
\begin{remark}\label{rem:probabilities}
To visualize the magnitude of the probabilities in Case I of \Cref{prop:fork}, we provide the plot of the cumulative function of the Beta$\(3,3733\)$ distibution in \Cref{fig:beta}. For $\mu=0.003$ (or 3/1000), the probability that the $\(\alpha,\mu\)$-chain finalises first is almost equal to $1$ (precisely equal to $0.999$). For lower $\mu$, the probability that this chain finalises first remains the largest. More precisely, this will be the case as long as 
\[\frac{\mu}{1-\mu}<\frac{3}{3733} \implies \mu<8.036\times 10^{-4}\]
which shows that if at least $\mu_0\approx 8.036\times 10^{-4}$ miners are working on the chain of the $\alpha\ge2/3$ validators, then this chain is more likely than not to finalise checkpoints first. We note, that such scenarios are extreme and only of theoretical interest. In the case that only such a small fraction of miners is working on a chain, then the $\alpha>2/3$ validators will be able to notice that something is \qt{wrong} due to the extremely slow block creation times. Hence, such a scenario is expected to be resolved with some form of off-chain communication. In any case, these cases remain only relevant (if at all) for the hybrid PoW-PoS structure.
\end{remark}

\begin{figure*}[!htb]
\centering
\includegraphics[scale=0.5]{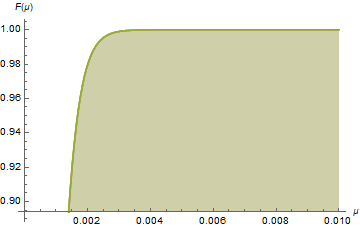}
\caption{Probability that the $\(\alpha,\mu\)$-chain with $\alpha>2/3$ validators will be able to finalize first. Even for values of $\mu$ as low as 0.003, this probability is almost equal to $1$. For any $\mu\ge 8.036\times 10^{-4}$, this probability is not lower than $0.5$.}
\label{fig:beta}
\end{figure*}

\begin{remark}[Intuition of \Cref{prop:fork}.] The argument in the proof of \Cref{prop:fork} is illustrated in \Cref{fig:safety}. After the initiation of the fork, validators who keep voting in the upper branch know that they will be able to start finalising first, since they form the majority at any point in time. Validators in the lower branch will also be able to finalise checkpoints, yet this will take considerably longer, see \Cref{fig:finalization}. In this case, conflicting checkpoints will be finalised and clients aware of either of the checkpoints will not be willing to revert them (under Casper's fork-choice rule). In other words, even if a double-finality event does take place, users are not forced to accept the claim that has more stake behind it; instead, users will be able to manually choose which fork to follow along, and are certainly able to simply choose \qt{the one that came first}. A successful attack in Casper looks more like a hard-fork than a reversion, and the user community around an on-chain asset is quite free to simply apply common sense to determine which fork was not an attack and actually represents the result of the transactions that were originally agreed upon as finalized. In any case, as shown in \Cref{prop:fork}, the time of a network partition that is required for the double-finalization event to occur, is a period of time which is of theoretical interest only. Finally, note that these bounds are conservative, i.e., they represent a worst case scenario, since they assume that all the adversarial stake is coordinated on the same chain (fork). In case, of $n>2$ conflicting forks, these times increase even more for any given percentage of benevolent validators that remain on the \qt{canonical} chain. 
\begin{figure*}[!hbt]
\centering
\includegraphics{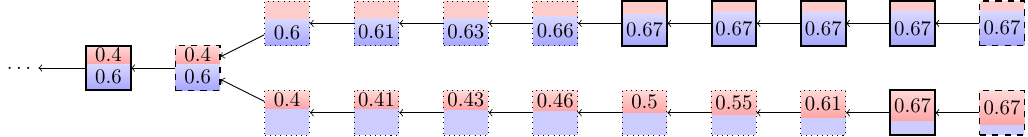}
\caption{Casper's safety property, cf. \Cref{prop:fork}: in case of a fork, the branch with the majority of validators (here upper branch) will resume finalising checkpoints first. Nodes do not revert finalised checkpoints, so for an honest node to accept a minority attacker's finalised block, it must have remained unaware of finalised blocks on the honest chain until the attacker was able to finalise, which for a $49\%$-strong attacker would take $2698 - 2546 = 152$ epochs or over 29 hours. After finalisation resumes, the ESF drops to $2$, which makes rapid changes in validators' deposits once again impossible, cf. \eqref{eq:rewardfactor}.}
\label{fig:safety}
\end{figure*}
\end{remark}

\subsection{Incentive Compatibility}\label{sec:incentivecomp}

In this section we will consider the system's \emph{incentive compatibility}, i.e., whether the reward scheme incentivises all types of nodes in the system to follow the protocol. We discuss the two main types of nodes, i.e., block proposers and validators, separately.

{\bf Block proposers:}
Currently, the incentives for block proposers on the PoW chain are similar to those in Bitcoin, for which it has been shown in \cite{eyal2018majority,sapirshtein2016optimal,kiayias2016blockchain} to be a Nash equilibrium to follow the protocol for anyone who controls less than roughly $\frac{1}{3}$ of the mining power. The checkpointing mechanism has a small impact on the profitability of selfish mining because it imposes limits on how many blocks can be reverted and when. After all, the new fork-choice rules prefers blocks that have the highest justified checkpoint, meaning that it is impossible for a selfish miner to convince other nodes to drop blocks that go beyond the last justified checkpoint. However, the epoch length (50 blocks) is likely to be too large for the impact to be considerable.

Block proposers also influence whether and when some Casper FFG protocol messages appear on the PoW chain (for {\tt vote} messages, the \qt{when} is very important because delaying them beyond the end of the epoch renders them invalid). Hence, the block proposers should have an incentive to include these messages. As mentioned in EIP 1011 \cite{ryan2018eip}, a miner who includes a vote message receives a fee equal to $\frac{1}{8}$ of the reward given to the validator (see also lines 297-301 of the Casper contract \cite{caspercontract}). Meanwhile, according to EIP 1011 \cite{ryan2018eip}, the contract calls do not consume gas, and therefore also do not contribute to the block's gas limit. Hence, block proposers are incentivised to include votes unless the computational cost incurred by the miner outweighs the reward for processing the vote message, which depends on the validator's deposit, and the protocol's total deposit.

%On the other hand, the impact that they have by obstructing messages is limited.
%As mentioned in Section~\ref{sec:model}, if $l$ is large enough, then even a small minority of honest miners can get votes onto the blockchain in time with high probability. For refusal to publish {\tt slash}, {\tt deposit} or {\tt withdraw} messages to have any effect, they must be kept off the blockchain (nearly) indefinitely, so as long as a clear majority of the block proposers is honest, the impact of the checkpointing mechanism on the incentives of the block proposers is limited.

{\bf Validators:}
The validators are the new class of nodes introduced by Casper FFG. Their most important action is to vote, so we will consider the impact on their deposits of producing a valid vote or not. In the following, we will consider the case of a validator $\vldr$ who controls a deposit share \mbox{$\vdsize \in (0,\frac{2}{3})$}. (If $\alpha > \frac{2}{3}$ then the validators have complete control over the protocol anyway.) The combined stake fraction of $\vldr$ and the other validators who vote is given by $\vtsize$ (the $1-\vtsize$ who do not vote could be doing so because they have unwillingly gone offline). We consider a single epoch, so we write $\rewardfactor$ instead of $\rewardfactor_i$. Both the voting and non-voting validators are hurt when $\vldr$ does not vote due to the dependence of the collective reward factor of \eqref{eq:rewardfactor} on the total fraction of voting validators. Furthermore, the impact depends on whether the epoch is justified, and whether $\vldr$ has an impact on that (if so, we call $\vldr$ a \qt{swing voter}). The impact of $\vldr$ not voting on all types of nodes and in all scenarios is summarised in Table~\ref{tab:losses}. It can be seen from the values in Table~\ref{tab:losses} that if $\vldr$ does not vote, then all validators lose a bit, but $\vldr$ loses by far the most. For example, if $\rho=10^{-6}$, $m=1$ and $\alpha=0.2$ then the losses for $\vldr$ are more than ten times as high as the losses for both the voting and non-voting validators. This by itself shows that $\vldr$ is incentivised to follow the protocol, if $\vldr$ is motivated by maximising her protocol rewards. Also, because all validators lose a bit, it does not pay in the short term for a validator to cause another validator to not vote, e.g., by censoring. 

\begin{table}[!htb]
\centering
\setlength{\tabcolsep}{18pt}
\renewcommand{\arraystretch}{1.2}
\begin{tabular}{@{}rccc@{}}
\toprule
&\multicolumn{3}{c}{\b Losses}\\
\cmidrule{2-4}
\b Justification & $\vldr$ & Voters & Non-voters\\
\cmidrule(r{2pt}){1-1}\cmidrule(l{2pt}){2-4}
Always & $\dfrac{\rewardfactor}{1+\rewardfactor}\(1+\dfrac{1}{2}\(\vtsize\rewardfactor + \vdsize\)\)$ & 
$\dfrac{1}{2} \vdsize \rewardfactor $ & $\dfrac{\rewardfactor}{1+\rewardfactor}\cdot\(\dfrac{1}{2} \vdsize\)$ \\[0.3cm]
Never & $\dfrac{\rewardfactor}{1+\rewardfactor}$ & 0 & 0 \\[0.3cm]
$\vldr$ is swing voter & $\dfrac{\rewardfactor}{1+\rewardfactor}\(1+\dfrac{1}{2}\(\vtsize\rewardfactor + \vdsize\)\)$ &$\dfrac{1}{2} \vtsize \rewardfactor$ & $\dfrac{\rewardfactor}{1+\rewardfactor}\cdot\(\dfrac{1}{2} \vtsize\)$\\
\bottomrule
\end{tabular}
\vspace{0.35cm}
\caption{Losses to all groups of validators if $\vldr$ abstains from voting. }
\vspace{-0.5cm}
\label{tab:losses}
\end{table}

For a validator $\vldr$ whose utilities go beyond the simple maximisation of their deposits, we can go a step further by providing upper bounds on the damage that $\vldr$ can do relative to $\vldr$'s own losses in a \emph{single epoch}. We consider two measures: the ratio between the \emph{absolute} losses of the other validators to $\vldr$'s, which we call the Griefing Factor (GF), and the ratio between the \emph{relative} losses, which we call the Proportional Loss Ratio (PLR)  \cite{buterin2018discouragement}.

In general, it holds that if the attacker controls $\vdsize$ of deposits, then 
$\text{GF} = \frac{1-\vdsize}{\vdsize} \text{PLR}$ \cite{buterin2018discouragement}.
Table~\ref{tb:abstain-gfs} displays the PLRs, obtained by dividing the relevant entries in Table~\ref{tab:losses} by each other and applying the equality above. All PLRs are bounded from above by $\frac{1}{2}$, and are much lower for small $\vdsize$. More concretely, if $\rewardfactor \approx 0$, then the PLRs are ${\frac{1}{2}\vdsize}/{(1+\frac{1}{2}\vdsize)}$, 0, and ${\frac{1}{2}\vtsize}/{(1+\frac{1}{2}\vtsize)}$ for voters in the three scenarios of Table~\ref{tb:abstain-gfs} respectively, so the PLR is always between $0$ (at $\vdsize\approx0$) and $\frac{1}{3}$ (at $\vdsize>\frac{1}{3}$ and $\vtsize=1$) whereas the GF is always between $1$ (at $\vdsize\approx0$) and~$\frac{1}{3}$ (at $\vdsize>\frac{1}{3}$ and $\vtsize=1$). This means that $\vldr$ will always incur a relative cost that is higher than the damage done to the others by abstaining. 

\begin{table}[!htb]
\centering
\setlength{\tabcolsep}{18pt}
\renewcommand{\arraystretch}{1.2}
\begin{tabular}{@{}rcccc@{}}
\toprule
& \multicolumn{2}{c}{\b Proportional Loss Ratio (PLR)} & \multicolumn{2}{c}{\b Griefing Factor (GF)}\\
\cmidrule(l{2pt}r{2pt}){2-3} \cmidrule(l{2pt}){4-5}
\b Justification & Voters & Non-voters & Voters & Non-voters \\ 
\cmidrule(r{2pt}){1-1}\cmidrule(l{2pt}r{2pt}){2-3} \cmidrule(l{2pt}){4-5}
Always & $\displaystyle\frac{\frac{1}{2} \vdsize (1+\rewardfactor)}{1+\frac{1}{2}(\vtsize\rewardfactor + \vdsize)} $ & $\displaystyle\frac{\frac{1}{2} \vdsize}{1+\frac{1}{2}(\vtsize\rewardfactor + \vdsize)} $ &
$\displaystyle\frac{\frac{1}{2} (\vtsize-\vdsize)(1+\rewardfactor)}{1+\frac{1}{2}(\vtsize\rewardfactor + \vdsize)} $ & $\displaystyle\frac{\frac{1}{2} (1-\vtsize)}{1+\frac{1}{2}(\vtsize\rewardfactor + \vdsize)} $ \\[0.3cm] 
Never & 0 & 0 & 0 & 0 \\[0.3cm]
$\vldr$ is swing voter & $\displaystyle\frac{\frac{1}{2} \vtsize \(1+\rewardfactor\)}{1+\frac{1}{2}\vtsize(1+\rewardfactor)} $ &  $\displaystyle\frac{\frac{1}{2} \vtsize}{1+\frac{1}{2}\vtsize(1+\rewardfactor)} $ &
$\displaystyle\frac{\frac{1}{2} \vtsize (\vtsize-\vdsize)(1+\rewardfactor)}{\vdsize(1+\frac{1}{2}\vtsize(1+\rewardfactor))} $ & $\displaystyle\frac{\frac{1}{2} \vtsize (1-\vtsize)}{\vdsize(1+\frac{1}{2}\vtsize(1+\rewardfactor))} $ \\
\bottomrule
\end{tabular}
\label{tab:griefingfactors}
\vspace{0.35cm}
\caption{PLRs and GFs if $\vldr$ abstains from voting.}
\label{tb:abstain-gfs}
\vspace{-0.5cm}
\end{table}

\section{Implementation}\label{sec:implementation}

%Implementation of hybrid Casper FFG started in late 2017.
In this section, we discuss implementation specifics of the Casper contract, in particular the financial cost in terms of rewards to validators in \Cref{sub:choice}, its impact on transaction throughput in \Cref{sec:messages}, and potential issues and limitations in \Cref{sub:limitations}.

%https://github.com/paritytech/parity-ethereum/issues/7162

%A higher base interest $\baseinterest$ means that the protocol becomes more expensive to operate, but also more decentralised as more people will be willing to deposit. Setting rewards too low also increases the possibility of \textit{discouragement attacks} \cite{buterin2018discouragement}, where an attacker either performs a censorship attacks or finds a way to cause high latency (e.g., by splitting the network), causing validators to lose money; the threat of this may encourage validators to leave or discourage them from joining.
%A higher base penalty factor $\basepenalty$ means improved liveness, i.e., faster recovery from a large number of validators going offline. However, higher penalties also mean bigger losses for validators during serious network partitions.
%A higher deposit size dependence $\depsizefactor$ means that validators can make a larger profit by performing censorship or DoS attacks against other validators. However, setting it too low means that the protocol does not automatically adjust the interest rate depending on how risky potential validators perceive depositing to be (an argument also made in \cite{choi2017casper}). 

\subsection{Impact of the Parameter Choice}
\label{sub:choice}
Finding appropriate values for each of $\baseinterest$, $\basepenalty$, and $\depsizefactor$ as discussed in Section~\ref{sub:rewards} involves a tradeoff. In general, to find appropriate values for $\baseinterest$, $\basepenalty$, and $\depsizefactor$, there are several factors to consider.\par
\noindent
Regarding $\baseinterest$:
\begin{itemize}[noitemsep]
\item Setting the reward too high makes the protocol expensive to operate.
\item Setting the reward too low means that fewer people will be willing to deposit, which in turn implies that the protocol is more centralized and less secure \cite{Oce20}.
\item Setting rewards too low also increases the possibility of \textit{discouragement attacks} \cite{buterin2018discouragement}, where an attacker either performs a censorship attacks or finds a way to cause high latency (e.g., by splitting the network), causing validators to lose money; the threat of this may encourage validators to leave or discourage them from joining.
\end{itemize}
Regarding $\basepenalty$:
\begin{itemize}[noitemsep]
\item Setting the penalty factor too high implies improved liveness, i.e., faster recovery from situations in which more than $1/3$ stop voting properly.
\item Setting the penalties too high, however, also implies higher risks for validators and bigger losses even if they accidentally are not able to stay online.
\end{itemize}
Regarding $\depsizefactor$:
\begin{itemize}[noitemsep]
\item Setting the deposit size dependence higher means that validators can make a larger profit by performing censorship or DoS attacks against other validators.
\item Setting the deposit size dependence too low means that the protocol does not automatically adjust the interest rate depending on how risky potential validators perceive depositing to be, as argued in \cite{choi2017casper}.
\item Users are looking for not just low issuance and high security, but also \textit{stability} of the level of issuance and security. Having rewards decrease as the total deposit size increases (i.e., setting $\depsizefactor$ high) accomplishes the former goal trivially, and the latter goal by ``trying harder'' to attract validators when the total deposit size is small. However, low values of $\depsizefactor$ mean that from the perspective of a single validator, the interest rate is stable (e.g., if you deposit when the interest rate is $5\%$, you know that this will not drop dramatically even if a large sum is deposited by other validators).
\end{itemize}

%\item Users are looking for not just low issuance and high security, but also \textit{stability} of the level of issuance and security. Having rewards decrease as the total deposit size increases (i.e., setting $\depsizefactor$ high) accomplishes the former goal trivially, and the latter goal by ``trying harder'' to attract validators when the total deposit size is small. However, low values of $\depsizefactor$ mean that from the perspective of a single validator, the interest rate is stable (e.g., if you deposit when the interest rate is $5\%$, you know that this will not drop dramatically even if a large sum is deposited by other validators).
\subsubsection*{Funding}
Rewards are paid to validators by the protocol. As discussed in \cite{ryan2018eip}, the Casper contract was planned to receive an initial amount of 1.5M newly created ETH after a hard fork (coinciding with the changes to the fork-choice rule used by the clients). If 10M ETH were deposited, then this would keep the contract funded for roughly 2 years.\footnote{Meanwhile, the mining reward was planned to be reduced substantially from 3ETH per block to 0.6ETH per block \cite{ryan2018eip,ryan7}.} This amount of ETH is intentionally kept limited to maintain an informal (i.e., dependent on further hard forks) deadline for the transition to full PoS.\footnote{A similar decision was made for the PoW difficulty, which has a built-in \qt{difficulty bomb} as a deadline for the transition to partial or full PoS, see also \href{https://www.cryptocompare.com/coins/guides/what-is-the-ethereum-ice-age}{Cryptocompare.com}.} If the contract has insufficient funds, validators still earn interest (as their deposits are kept as contract variables), but are unable to withdraw their deposits.

\subsection{Overhead of the Hybrid Casper FFG Contract}
\label{sec:messages}

Based on the discussion of the previous section, it is now immediate that the calls to the Casper contract impact the throughput of the protocol because they use the same client bandwidth and processing power as regular transactions. To understand this extent of this impact, we will study the contract's consumption of gas (which, as discussed in Section~\ref{sub:powchain}, measures the computational load) relative to the total gas limit. The total block gas limit is variable and can be influenced by the miners, although since December 2017 it has been close to 8M gas (see \href{https://etherscan.io/chart/gaslimit}{Etherscan.io}). The estimated gas costs of the six main types of function calls in the contract are displayed in Table~\ref{tbl:functiongas}. Although the exact gas consumption of function calls depends on various external factors,  including the exact numerical values of the arguments, the functionality to produce estimates is built into the Vyper compiler.
\begin{table}[!ht]
\vspace{0cm}\centering
\begin{tabular}{lcc}
	\midrule
	Function & Estimated gas cost \\ \midrule
	{\tt initialize\_epoch} & 742\,393 \\
	{\tt deposit} & 831\,687 \\
	{\tt logout} & 131\,308 \\
	{\tt withdraw} & 224\,155 \\
	{\tt vote} & 532\,031 \\
	{\tt slash} & 280\,864 \\
	\midrule
\end{tabular} 
\vskip0.1cm
\caption{Gas costs for the six core functions of the Casper contract as estimated by the Vyper compiler.} \label{tbl:functiongas}
\end{table}
Per epoch, there is ideally one {\tt initialize\_epoch} call, and one {\tt vote} call per validator. The cost of the {\tt deposit}, {\tt logout} and {\tt withdraw} functions is assumed to be negligible, in part because of the minimum time period between depositing and withdrawing. We assume the same for {\tt slash} because of the high cost of violating a slashing condition. The load of the two other calls is unevenly distributed across the epoch: the {\tt initialize\_epoch} call will come early in the epoch, but most of the {\tt vote} calls are expected to arrive in the later part of the epoch when the probability of voting for a `losing' block in a temporary PoW chain fork is small enough. In the benchmark parametrisation, there are 50 blocks per epoch, and we assume that votes arrive in the final three quarters of the epoch, i.e., the last 37 blocks. The impact of the {\tt initialize\_epoch} call during the first 13 blocks is roughly equal to $742\textnormal{K} / (13 \cdot 8\textnormal{M}) \approx 0.7\%$, which is small. However, the impact of the votes during the last 37 blocks can be considerable. With 100 validators, the expected gas cost per epoch is roughly equal to $100 \cdot 532\textnormal{K} / (37 \cdot 8\textnormal{M}) \approx 18\%$. This confirms similar observations by Ethereum researchers that, even under proper protocol updates, no more than 592 (or even 400) validators could be supported \cite{ryan2018gas}. The implementation of Casper FFG as a smart contract is illustrated in \Cref{fig:singlechain}.
\begin{figure}[!htb]
\centering
\includegraphics{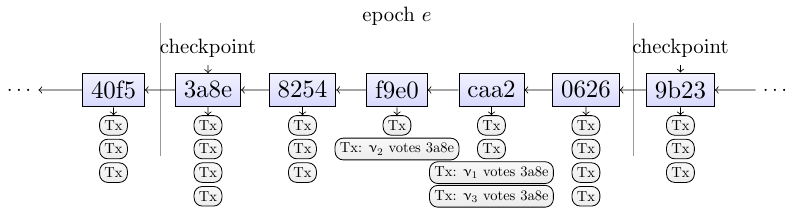}
\caption{Original set-up of Casper FFG: protocol messages are transactions to the Casper contract that appear on the main chain. In the above figure, the epoch length is 5, and three validators --- $\vldr_1$, $\vldr_2$, and $\vldr_3$ --- broadcast vote messages. Assuming that they control at least $\frac{2}{3}$ of the stake, epoch $\epoch$ is justified after the 4th block of $\epoch$, and if epoch $\epoch-1$ was also justified, then $\epoch-1$ is finalised.}
\label{fig:singlechain}
\end{figure}
Several approaches can be taken to limit the number of participating validators. The intended approach by Ethereum was to impose a fixed minimum deposit size of 1500 ETH. Alternative approaches would be to not accept new deposits beyond a hard limit of $N$ validators, to only accept votes from the $N$ validators with the highest deposit size, or to dynamically adjust the minimum deposit size based on the number of validators. Accurate predictions of the impact of the minimum deposit on the number of validators require economic modelling that is outside the scope of this paper. As for other PoS-based blockchain platforms: in EOS, 21 delegates \cite{EOS} chosen by the stakeholders control the consensus algorithm, whereas Cardano aims for 100 stake pools \cite{Cardano,brunjes2018reward}.

\subsubsection*{Off-Chain Messages}\label{sub:messages}
 Another approach to mitigate the network load is to move hybrid Casper FFG messages onto a separate chain. As a result, two interdependent blockchains operate simultaneously: the traditional PoW chain, and a side chain called the \emph{beacon chain}.\footnote{The beacon chain is named for its originally envisaged role in producing a random beacon \cite{drake2018which,edgington2018}.} The core protocol messages ({\tt vote} and {\tt slash}) are then moved to the latter, and the evolution of the rewards is processed internally by clients. A contract on the main chain is still created to handle {\tt deposit}, {\tt logout}, and {\tt withdraw} messages, and to process exchanges from ETH to/from the reward variables on the beacon chain.  The {\tt initialize\_epoch} calls are no longer necessary as clients process the epoch transitions internally.\par
The advantages of this approach are the possibility of message processing optimisations (e.g., bit masks and signature aggregation \cite{buterin2018casperffg,coredevs40meeting}, or the parallel processing of {\tt vote} messages which was found to be challenging in the contract set-up  \cite{Ch18}), and facilitation of a transition to a sharded\footnote{In a sharded blockchain, the network load is divided amongst several subchains that work in parallel.} blockchain, by serving as the central chain connecting the various side chains. The disadvantage is that substantial modifications to the clients will be required. The block proposal/consensus mechanism on the beacon chain is still under active development -- it is envisaged to use full PoS (with the same validator set as Casper FFG) in its final iteration. Given its long-term benefits, the dual-chain approach has been chosen as the way forward for Ethereum \cite{edgington2018}.

\subsection{Other Issues \& Limitations}
\label{sub:limitations}
We conclude with remarks of a general nature and issues that we detected from our study and the implementation of the Casper FFG contract. First, the \qt{finder's fee}, cf. \Cref{sub:scheme}, for detecting a violator of slashing conditions may create conflicting incentives and competition between validators.  Second, in the case that the network experiences a large partition or fork, cf. \Cref{fig:safety}, honest validators who have voted and finalised on the \emph{non-canonical} chain will sustain heavy losses to return to the main chain. 
%Also, newly entering validators may flip the majority between long-lasting forks. While these are mostly issues of theoretical interest or issues that can addressed with appropriate adjustments in a future specification as mentioned in \Cref{sec:analysis}, we mention them here for the sake of completeness. 
Third, to focus on validators' mechanics, we ignored potential collusions between validators and PoW miners. This point is not relevant in a pure PoS implementation, as currently put forward, and would have shifted the present analysis away from Casper's properties of interest. The way that new (entering) nodes will choose between conflicting checkpoints depends on the choice of bootstrapping nodes by the client, and is therefore a question of proper client implementation. Finally, despite increasing security, the checkpointing mechanism does not reduce confirmation times (2 epochs = 100 blocks). Instead, the full benefits of Casper in terms of block-confirmation times are expected to be realised in a pure PoS implementation of the Ethereum blockchain. 

%\subsection{Development History}
%
%Development of the contract spanned from roughly February 2017, when the first version of the 
%
%The first version of the Casper contract was uploaded to GitHub in February 2017, with s mostly stable version emerging towards the end of the year. In December 2017, a dedicated test net for hybrid Casper FFG was launched that allowed nodes using the Pyethereum client to interact with the contract \cite{Pe17}.
%In May 2018,  the Parity client launched its own test net \cite{Pe18}, and in June 2018 it was reported that Geth was close to also having an FFG- compliant version of their client ready on a test net \cite{Ol18}.

%See Appendix~\ref{app:beaconchain} for more detail.

%By contrast, the dual-chain approach allows for the utilisation of message processing optimisations (e.g., bit masks and signature aggregation \cite{buterin2018casperffg,coredevs40meeting}) as the messages are processed natively by the clients. The gains are expected to be so large that is considered possible to lower the minimum deposit size to 32ETH \cite{coredevs40meeting} (and possibly make it fixed), with a theoretical upper bound of $2^{22} \approx 4.2$ million validators \cite{buterin2018shasper}. In this case, even the {\tt deposit} and {\tt withdraw} messages may impose a considerable load, although again optimisations are available \cite{buterin2018casperffg}.

\section{Conclusions}\label{sec:conclusions}

In this paper, we analysed the Casper FFG contract that was evaluated on a dedicated Ethereum testnet. We described its core mechanism and showed that its incentives scheme ensures liveness whilst providing security against the finalisation of conflicting histories, i.e., checkpoints. As a finality protocol that can be overlaid on both PoW and PoS blockchains, hybrid Casper FFG can be of interest to a broad audience within the blockchain ecosystem. Our findings on liveness, safety, incentive compatibility, and implementation remain particularly relevant for Ethereum's transition to a sharded design in which the Casper FFG philosophy is carried over. \par
Although the focus of this paper was on the incentives scheme implemented in the Casper contract, an interesting direction for future work would be to explore different schemes, in particular alternatives to equations \eqref{eq:rewardfactor}, \eqref{eq:collective}, and \eqref{eq:interest}. All schemes will involve a trade-off: although heavy penalties for offline validators improve liveness and safety, they also negatively affect honest participants who suffer from unavoidable network or computer failures. Also, the impact of dynamic validator sets are an interesting research direction. 
%For applications of Casper FFG as an overlay on top of a PoW chain, it would also be interesting to evaluate collusions between adversarial miners and validators.
%Finally, development of `fully PoS'
%Casper FFG remains an active research direction.

\section*{Acknowledgements}
The authors would like to thank Pieter Hartel, Chih-Cheng Liang, Ivan Homoliak, and Vi for their helpful comments.

\bibliographystyle{abbrv}
\bibliography{casper_bib}
\end{document}